\documentclass[10pt,conference,letterpaper]{IEEEtran}

\usepackage{balance}
\usepackage{times}
\usepackage{amsmath}
\usepackage{amssymb}
\usepackage{epsfig}
\usepackage[thmmarks]{ntheorem}
\usepackage{subfigure}
\usepackage{url}
\usepackage{verbatim}
\usepackage{algorithmic}
\usepackage{multirow}
\usepackage[para]{footmisc}

\makeatletter
\newif\if@restonecol
\makeatother

\usepackage[ruled, vlined]{algorithm2e}

\newcommand{\eat}[1]{}

\newcommand{\reminder}[1]{ (((\mbox{$\Longleftarrow \star$}{\textbf{#1}} )))}

\newcommand{\eatreminders}[0]{\renewcommand{\reminder}[1]{}}

\newcommand{\system}[1]{{\ensuremath {\mathsf{#1}}}\xspace}

\newcommand{\QED}{\mbox{\rule[0pt]{1.2ex}{1.2ex}}}
\qedsymbol{\QED}
\theoremsymbol{\QED}  % buggy

\theorembodyfont{\normalfont\upshape}
\theoremheaderfont{\normalfont\bfseries}

{
\theoremstyle{plain}
\newtheorem{definition}{Definition}

\newtheorem{theorem}{Theorem}
\newtheorem{lemma}{Lemma}

\newtheorem{example}{Example}
\newtheorem{property}{Property}
\newtheorem*{proof}{Proof}
}

\newcommand{\topk}{top-\emph{k}}

\newcommand{\entity}[1]{\textsf{\scriptsize #1}}
\newcommand{\etuple}[2]{$\langle \entity{#1},\entity{#2}\rangle$}

\newcommand{\edge}[1]{\textsf{\emph{\scriptsize #1}}}
\newcommand{\edgeends}[2]{(\entity{#1}, \entity{#2})}

\newcommand{\eg}{e.g.,\xspace}
\newcommand{\ie}{i.e.,\xspace}
\newcommand{\spara}[1]{\smallskip\noindent{\bf #1}}

\newcommand{\uf}{\mathcal{UF}}
\newcommand{\lf}{\mathcal{LF}}
\newcommand{\ub}{\mathcal{UB}}
\newcommand{\nb}{\mathcal{NB}}
\newcommand{\lattice}{\mathcal{L}}

\newcommand{\ubs}{U}

\newcommand{\supgraph}{\succeq}

\eatreminders

\title{Querying Knowledge Graphs by Example Entity Tuples}

\author{
{Nandish Jayaram$^\dag$\hspace*{0.5em}
Arijit Khan$^\S$\hspace*{0.5em}
Chengkai Li$^\dag$\hspace*{0.5em}
Xifeng Yan$^\S$\hspace*{0.5em}
Ramez Elmasri$^\dag$
}
\vspace{1.6mm}\\
\fontsize{10}{10}\selectfont\itshape
$^\dag$University of Texas at Arlington, $^\S$University of California, Santa Barbara
}

\begin{document}
\maketitle
\thispagestyle{plain}
\pagestyle{plain}
\begin{abstract}
We witness an unprecedented proliferation of knowledge graphs that record millions of entities and their relationships.  While knowledge graphs are structure-flexible and content-rich, they are difficult to use.  The challenge lies in the gap between their overwhelming complexity and the limited database knowledge of non-professional users.  If writing structured queries over ``simple'' tables is difficult, complex graphs are only harder to query.  As an initial step toward improving the usability of knowledge graphs, we propose to query such data by example entity tuples, without requiring users to form complex graph queries.  Our system, \system{GQBE} (Graph Query By Example), automatically derives a weighted hidden maximal query graph based on input query tuples, to capture a user's query intent.  It efficiently finds and ranks the top approximate answer tuples.  For fast query processing, \system{GQBE} only partially evaluates query graphs.  We conducted experiments and user studies on the large Freebase and DBpedia datasets and observed appealing accuracy and efficiency. Our system provides a complementary approach to the existing keyword-based methods, facilitating user-friendly graph querying. To the best of our knowledge, there was no such proposal in the past in the context of graphs.
\end{abstract}

\section{Introduction}\label{sec:intro}

\reminder{mention Facebook entity graph, semantic graph, SPARQL, we ignore numeric values, some QBE-related work on RDF...}

\reminder{related work
SWiPE: Searching Wikipedia by Example
RDF-QBE: a Semantic Web Building Block
Semantic Query-by-Example for RDF data - Sogang University
}

There is an unprecedented proliferation of \emph{knowledge graphs} that record millions of entities (e.g., persons, products, organizations) and their relationships. Fig.\ref{fig:example-graph} is an excerpt of a knowledge graph, in which the edge labeled \edge{founded} between nodes \entity{Jerry Yang} and \entity{Yahoo!} captures the fact that the person is a founder of the company.  Examples of real-world knowledge graphs include DBpedia~\cite{AuerBK+07}, YAGO~\cite{SuchanekKW07}, Freebase~\cite{Bollacker+08freebase} and Probase~\cite{probase}.
Users and developers are tapping into knowledge graphs for numerous applications, including search, recommendation and business intelligence.
%%%Google's Knowledge Graph, for instance, enhances users' search experience using data from Freebase and other sources.
%%%It echoes the shift of Web users' interest towards entity-related information, evidenced by the estimation that 71\% of Microsoft Bing search queries contain named entities~\cite{Yin+10};

%\begin{figure*}[htb]
%\begin{minipage}[b]{0.48\linewidth}
%\centering
%  \includegraphics[width = 0.95\linewidth, keepaspectratio = true, scale=0.40]{figures/erGraph.eps}%\vspace{-5mm}
%\caption{An Excerpt of a Knowledge Graph}
%\label{fig:example-graph}
%\end{minipage}%\vspace{-3mm}
%\begin{minipage}[b]{0.48\linewidth}
%\centering
%  \includegraphics[width = 0.88\linewidth, keepaspectratio = true, scale=0.40]{figures/maxQueryGraph.eps}%\vspace{-5mm}
%\caption{Neighborhood Graph for \etuple{Jerry Yang}{Yahoo!}}
%\label{fig:neighbor-graph}
%\end{minipage}%\vspace{-3mm}
%\end{figure*}

Both users and application developers are often overwhelmed by the daunting task of understanding and using knowledge graphs.  This largely has to do with the sheer size and complexity of such data.  As of March 2012, the Linking Open Data community had interlinked over 52 billion RDF triples spanning over several hundred datasets.  More specifically, the challenges lie in the gap between complex data and non-expert users.
Knowledge graphs are often stored in relational databases, graph databases and triplestores (cf.~\cite{KhanWY12} for a tutorial).
%%%While a number of graph database systems, triplestores, and RDF stores have emerged in recent years (e.g., Pregel~\cite{pregel}, GBase~\cite{gbase}, Trinity~\cite{trinity} and many others), \emph{usability} has not been the focus of innovation.
\begin{figure}[htb]
\centering
  \includegraphics[width = 0.99\linewidth, keepaspectratio = true, scale=0.4]{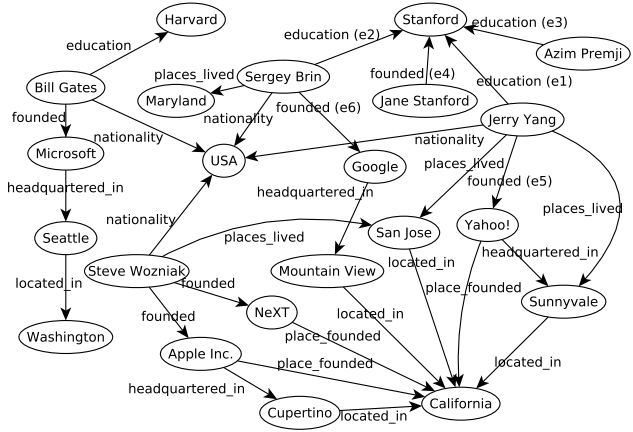}\vspace{-4mm}
\caption{An Excerpt of a Knowledge Graph}
\label{fig:example-graph}
\end{figure}%\vspace{-3mm}

\begin{figure}[tb!]
\centering
  \includegraphics[width = 0.9\linewidth, keepaspectratio = true, scale=0.4]{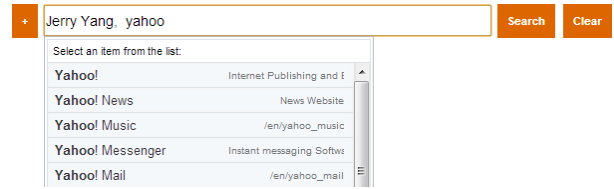}\vspace{-3mm}
\caption{Query Interface of \system{GQBE}}
\label{fig:searchScreen}\vspace{1mm}
\end{figure}
In retrieving data from these databases, the norm is often to use structured query languages such as SQL, SPARQL, and those alike.  However, writing structured queries requires extensive experiences in query language and data model and good understanding of particular datasets~\cite{usability}.  
%%%For this reason, database usability has received considerable attention lately (cf. an overview in~\cite{usability}).  
Graph data is not ``easier'' than relational data in either query language or data model.  The fact it is schema-less makes it even more intangible to understand and query. \emph{If querying ``simple'' tables is difficult, aren't complex graphs harder to query?} \reminder{CL: Explain why it is difficult to write a structured query on data graph.}

%\begin{figure}[tb!]
%\centering
%  \includegraphics[width = 1.0\linewidth, keepaspectratio = true, scale=0.4]{figures/searchScreen.eps}%%\vspace{-4mm}
%\caption{Search Screen of \system{GQBE}}
%\label{fig:searchScreen}%%\vspace{-2mm}
%\end{figure}

\begin{figure*}[htb]
\centering
  \includegraphics[keepaspectratio = true, scale=0.47]{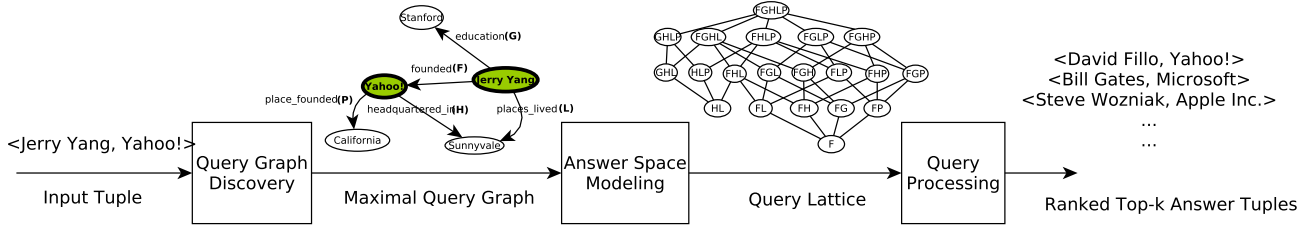}%\vspace{-4mm}
\caption{The Architecture and Components of \system{GQBE}}
\label{fig:architecture}%\vspace{-3mm}
\end{figure*}

Motivated by the aforementioned usability challenge, we build \system{GQBE}~\footnote{A description of \system{GQBE}'s user interface and demonstration scenarios can be found in~\cite{gqbedemo}.  An accompanying demonstration video is at {\small \url{http://www.youtube.com/watch?v=4QfcV-OrGmQ}}.} (Graph Query by Example), a system that queries knowledge graphs by example entity tuples instead of graph queries.  Given a data graph and a query tuple consisting of entities, \system{GQBE} finds similar answer tuples.  Consider the data graph in Fig.\ref{fig:example-graph} and an
scenario where a Silicon Valley business analyst is interested in finding entrepreneurs who have founded technology companies head-quartered in California.
Suppose she knows an example 2-entity query tuple such as \etuple{Jerry Yang}{Yahoo!} that satisfies her query intent.
As the query interface in Fig.~\ref{fig:searchScreen} shows, entering such an example tuple to \system{GQBE} is simple, especially with the help of user interface
tools such as auto-completion in identifying the exact entities in the data graph.
The answer tuples can be \etuple{Steve Wozniak}{Apple Inc.} and \etuple{Sergey Brin}{Google}, which are founder-company pairs.
If the query tuple consists of 3 or more entities (e.g., $\langle \entity{Jerry Yang},\entity{Yahoo!},$ $\entity{Sunnyvale}\rangle$), the answers will be similar tuples of the same cardinality (e.g., $\langle \entity{Steve Wozniak},\entity{Apple Inc.},\entity{Cupertino}\rangle$).
%Motivated by the aforementioned usability challenge, we build \system{GQBE} (Graph Query by Example), a system that queries knowledge graphs by example entity tuples instead of graph queries.  Given a data graph and a query tuple consisting of entities, \system{GQBE} finds similar answer tuples.  Consider the data graph in Fig.\ref{fig:example-graph}.  For a 2-entity query tuple \etuple{Jerry Yang}{Yahoo!}, the answer tuples can be \etuple{Steve Wozniak}{Apple Inc.}, \etuple{Sergey Brin}{Google} and \etuple{Bill Gates}{Microsoft}, which are all founder-company pairs. If the query tuple consists of 3 or more entities (e.g., $\langle \entity{Jerry Yang},\entity{Yahoo!},$ $\entity{San Jose}\rangle$), the answers will be similar tuples of the same cardinality (e.g., $\langle \entity{Steve Wozniak},\\
%\entity{Apple Inc.},\entity{San Jose}\rangle$).
%%\etriple{Steve Wozniak}{Apple Inc.}{San Jose}).%%%, since apart from the founder-company relationship, both \entity{Jerry Yang} and \entity{Steve Wozniak} have the relationship \edge{places\_lived} with \entity{San Jose}).

Our work is the first to query knowledge graphs by example entity tuples. The paradigm of \emph{query-by-example} (QBE) has a long history in relational databases~\cite{qbe}.   The idea is to express queries by filling example tables with constants and shared variables in multiple tables, which correspond to selection and join conditions, respectively.  Its simplicity and improved user productivity make QBE an influential database query language.  By proposing to query knowledge graphs by example tuples, our premise is that the QBE paradigm will enjoy similar advantages on graph data. The technical challenges and approaches are vastly different, due to the fundamentally different data models.

Substantial progress has been made on query mechanisms that help users construct query graphs or even do not require explicit query graphs.  Such mechanisms include keyword search (e.g., \cite{KA11}), keyword-based query formulation~\cite{PoundIW10,YCHH12}, natural language questions~\cite{Yahya+12}, interactive and form-based query formulation~\cite{Demidova+12,Jarrar+12}, and visual interface for query graph construction~\cite{GRAPHITE, GBLENDER}.  Little has been done on comparison across these graph query mechanisms.  While a usability comparison of these mechanisms and \system{GQBE} is beyond the scope of this paper, we note that they all have 
pros and cons and thus complement each other. 

%%%\emph{Query-by-example} provides an alternative approach to existing keyword-based methods, both facilitating user-friendly graph querying.  These two techniques have their own application scenarios.
Particularly, QBE and keyword-based methods are adequate for different usage scenarios.  
Using keyword-based methods, a user has to articulate query keywords, e.g., ``technology companies head-quartered in California and their founders'' for the aforementioned analyst.  
Not only a user may find it challenging to clearly articulate a query, but also a query system might not return accurate answers, since it is non-trivial to precisely separate these keywords and correctly match them with entities, entity types and relationships.
This has been verified through our own experience on a keyword-based system adapted from SPARK~\cite{spark}. 
In contrast, a \system{GQBE} user only needs to know the names of some entities in example tuples, without being required to specify how exactly the entities are related. 
%%%As the starting point of \system{GQBE}, matching names ``Jerry Yang'' and ``Yahoo!'' to entities is much easier.
%%%They are far more selective and less ambiguous than ordinary keywords.  
On the other hand, keyword-based querying is more adequate when a user does not know a few sample answers with respect to her query.

In the literature on graph query, the input to a query system in most cases is a
structured query, which is often graphically presented as a query graph or
pattern. Such is not what we refer to as query-by-example, because the query graphs and patterns are formed by using structured query
languages or the aforementioned query mechanisms. For instance, \system{PathSim}~\cite{Sun+11} finds the \topk\ similar entities that are connected to a query entity, based on a user-defined meta-path semantics in a heterogeneous network. In~\cite{YSZH12}, given a query graph as input, the system finds structurally isomorphic answer graphs with semantically similar entity nodes.
%In both works, a user should know the network schema to specify a meta-path or a query graph.
In contrast, \system{GQBE} only requires a user to provide an entity tuple, without knowing the underlying schema.
%%%Due to the very different data and query model, the work is a departure from prior study on database usability that mainly focuses on %%%relational databases.

%%%Query graphs or patterns that present queries over graphs are not what we refer to as query-by-example.  They are formed by using structured query languages or other query mechanisms such as keyword query~\cite{PoundIW10,YCHH12} and interactive query formulation~\cite{Demidova+12,Jarrar+12,GRAPHITE,GBLENDER}.

There are several challenges in building \system{GQBE}. Below we provide a brief overview of our approach in tackling these challenges. The ensuing discussion refers to the system architecture and components of \system{GQBE}, as shown in Fig.~\ref{fig:architecture}.

(1) With regard to \emph{query semantics}, since the input to \system{GQBE} is a query tuple instead of an explicit query graph, the system must derive a hidden query graph based on the query tuple, to capture user's query intent.  The \emph{query graph discovery} component (Sec.\ref{sec:qgraph}) of \system{GQBE} fulfills this requirement and the derived graph is termed a \emph{maximal query graph} (MQG).
The edges in MQG, weighted by several frequency-based and distance-based heuristics, represent important ``features'' of the query tuple to be matched in answer tuples.  More concretely, they capture how entities in the query tuple (i.e., nodes in a data graph) and their neighboring entities are related to each other.  Answer graphs matching the MQG are projected to answer tuples, which consist of answer entities corresponding to the query tuple entities.  %%%For an answer tuple, its answer entities and their neighboring entities are thus expected to be also related in ways captured by the MQG.
\system{GQBE} further supports multiple query tuples as input which collectively better capture the user intent.

(2) With regard to \emph{answer space modeling} (Sec.\ref{sec:modeling}), there can be a large space of approximate answer graphs (tuples), since it is unlikely to find answer graphs exactly matching the MQG.
%%%%In the aforementioned example, \etuple{Steve Wozniak}{Apple Inc.} might be a more accurate answer tuple than others, because both \entity{Steve Wozniak} and \entity{Jerry Yang} have lived in \entity{San Jose}, and both \entity{Apple Inc.} and \entity{Yahoo!} were founded and are headquartered in \entity{California}. $\langle \entity{Sergey Brin}$, $\entity{Google}\rangle$ is arguably a better answer than \etuple{Bill Gates}{Microsoft} since \entity{Google} is also headquartered in \entity{California}, and both \entity{Jerry Yang} and \entity{Sergey Brin} are graduates of \entity{Stanford}. %%%, although they both might be weaker answers than $\langle \entity{Steve Wozniak}$, $\entity{Apple Inc.}\rangle$.
%%%Even with regard to the edges labeled \edge{founded}, it is not necessary that the two entities in an query tuple must bear only such an edge between them.
%%%With regard to the relationships between query entities, there can be multiple direct and indirect paths between them. For instance, \entity{Jerry Yang} and \entity{Yahoo!} are connected through \entity{San Jose} and \entity{California}. The entities in query and answer tuples only need to have some paths in common.
\system{GQBE} models the space of answer tuples by a \emph{query lattice} formed by the subsumption relation between all possible query graphs.  Each query graph is a subgraph of the MQG and contains all query entities.  Its answer graphs are also subgraphs of the data graph and are isomorphic to the query graph.  Given an answer graph, its entities corresponding to the query tuple entities form an answer tuple. %%%%, according to the bijection between the entity sets of the query graph and the answer graph.
Thus the answer tuples are essentially approximate answers to the MQG.  For ranking answer tuples, their scores are calculated based on the edge weights in their query graphs and the match between nodes in the query and answer graphs.

(3) The query lattice can be large.  To obtain \topk\ ranked answer tuples, the brute-force approach of evaluating all query graphs in the lattice can be prohibitively expensive.
For \emph{efficient query processing} (Sec.\ref{sec:processing}), \system{GQBE} employs a top-$k$ lattice exploration algorithm that only partially evaluates the lattice nodes in the order of their corresponding query graphs' upper-bound scores.%\vspace{1mm}

We summarize the contributions of this paper as follows:%\vspace{-1mm}
\begin{list}{$\bullet$}
{ \setlength{\leftmargin}{0.5em} \setlength{\itemsep}{0pt}}
\item For better usability of knowledge graph querying systems,  we propose a novel approach of querying by example entity tuples, which saves users the burden of forming explicit query graphs. To the best of our knowledge, there was no such proposal in the past. %%%Our system \system{GQBE} tackles the challenges in supporting this query approach.
\item The query graph discovery component of \system{GQBE} derives a hidden maximal query graph (MQG) based on input query tuples, to capture users' query intent. \system{GQBE} models the space of query graphs (and thus answer tuples) by a query lattice based on the MQG.
\item \system{GQBE}'s efficient query processing algorithm only partially evaluates the query lattice to obtain the \topk\ answer tuples ranked by how well they approximately match the MQG.
\item We conducted extensive experiments and user study on the large Freebase and DBpedia datasets to evaluate \system{GQBE}'s accuracy and efficiency (Sec.\ref{sec:exp}). The comparison with a state-of-the-art graph querying framework \system{NESS}\cite{ness} (using MQG as input) shows that \system{GQBE} is twice as accurate as \system{NESS} and outperforms \system{NESS} on efficiency in most of the queries.
\end{list}

%%%% PAPER OUTLINE %%%%
%%%The rest of the paper is organized as follows.
%%%Sec.\ref{sec:prelim} presents the data model, the query model, and the problem statement.  We describe the query graph discovery component of \system{GQBE} in Sec.\ref{sec:qgraph}, modeling of the answer space in Sec.\ref{sec:modeling} and the query processing component in Sec.\ref{sec:processing}. Sec.\ref{sec:exp} presents experimental evaluation results and Sec.\ref{sec:related} discusses related work.  In Sec.\ref{sec:conclude}, we conclude this work.

\section{Problem Formulation}
\label{sec:prelim}

\system{GQBE} runs queries on knowledge data graphs.
A \emph{\textbf{data graph}} is a directed multi-graph $G$ with node set $V(G)$ and edge set $E(G)$.
Each node $v$$\in$$V(G)$ represents an entity and has a unique identifier $id(v)$.~\footnote{Without loss
of generality, we use an entity's name as its identifier in presenting
examples, assuming entity names are unique.}
Each edge $e$=$(v_i,v_j)$$\in$$E(G)$ denotes a directed relationship from
entity $v_i$ to entity $v_j$. It has a label, denoted as $label(e)$.
Multiple edges can have the same label.
The user input and output of \system{GQBE} are both entity tuples,
called \emph{\textbf{query tuples}} and \emph{\textbf{answer tuples}}, respectively.
A tuple $t$=$\langle v_1, \ldots, v_n \rangle$ is an ordered list of
entities (i.e., nodes) in $G$.  The constituting entities of query (answer)
tuples are called \emph{query (answer) entities}.  Given a data graph $G$
and a query tuple $t$, our goal is to find the \topk\ answer tuples $t'$
with the highest similarity scores $\textsf{score}_t(t')$.

\begin{figure}[t]
\centering
  \includegraphics[width = 0.88\linewidth, keepaspectratio = true, scale=0.40]{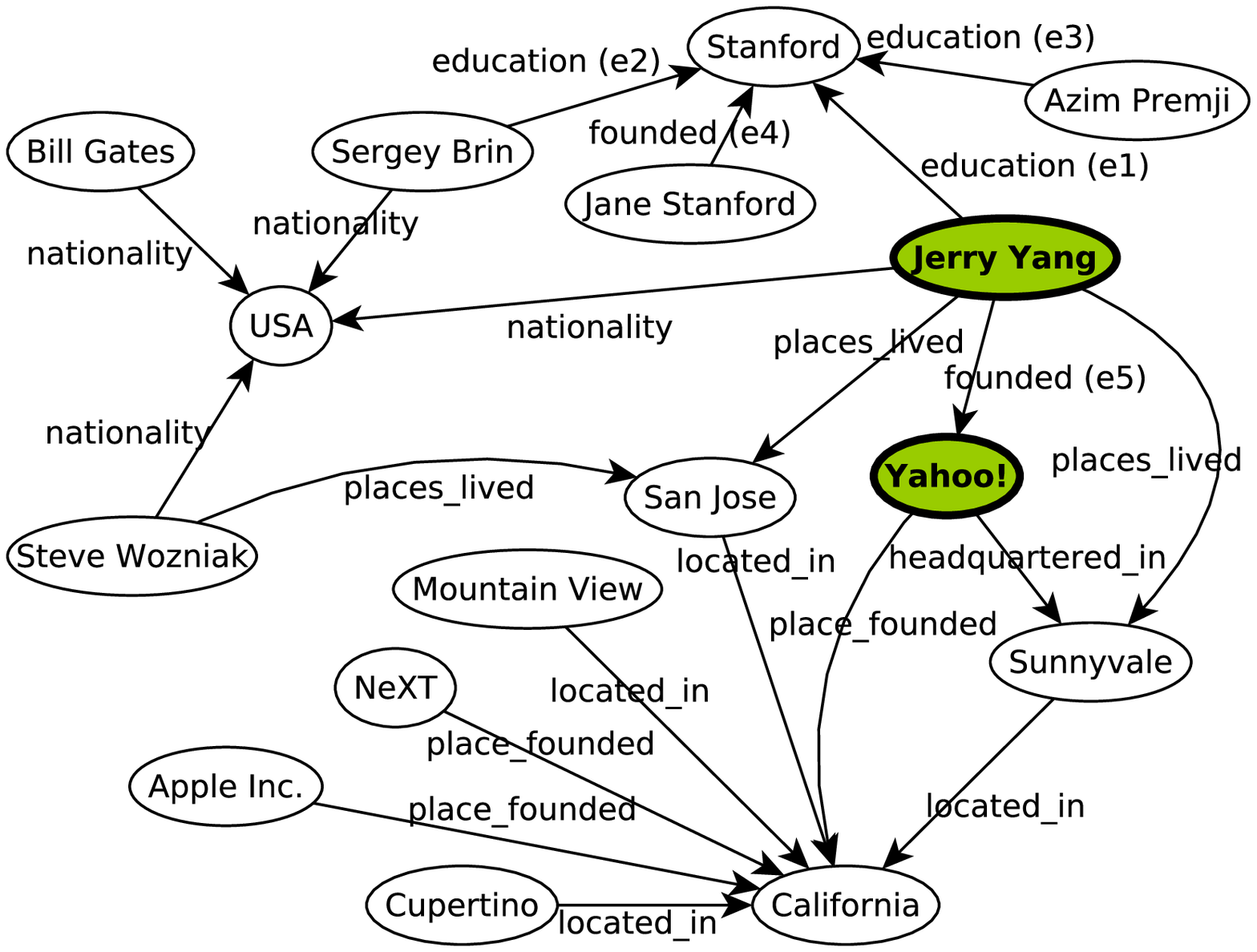}%%\vspace{-5mm}
\caption{Neighborhood Graph for \etuple{Jerry Yang}{Yahoo!}}
\label{fig:neighbor-graph}
\end{figure}

\begin{figure*}[htb]
\begin{minipage}[b]{0.33\linewidth}
\centering
  \includegraphics[width = 1.0\linewidth, keepaspectratio = true, scale=0.6]{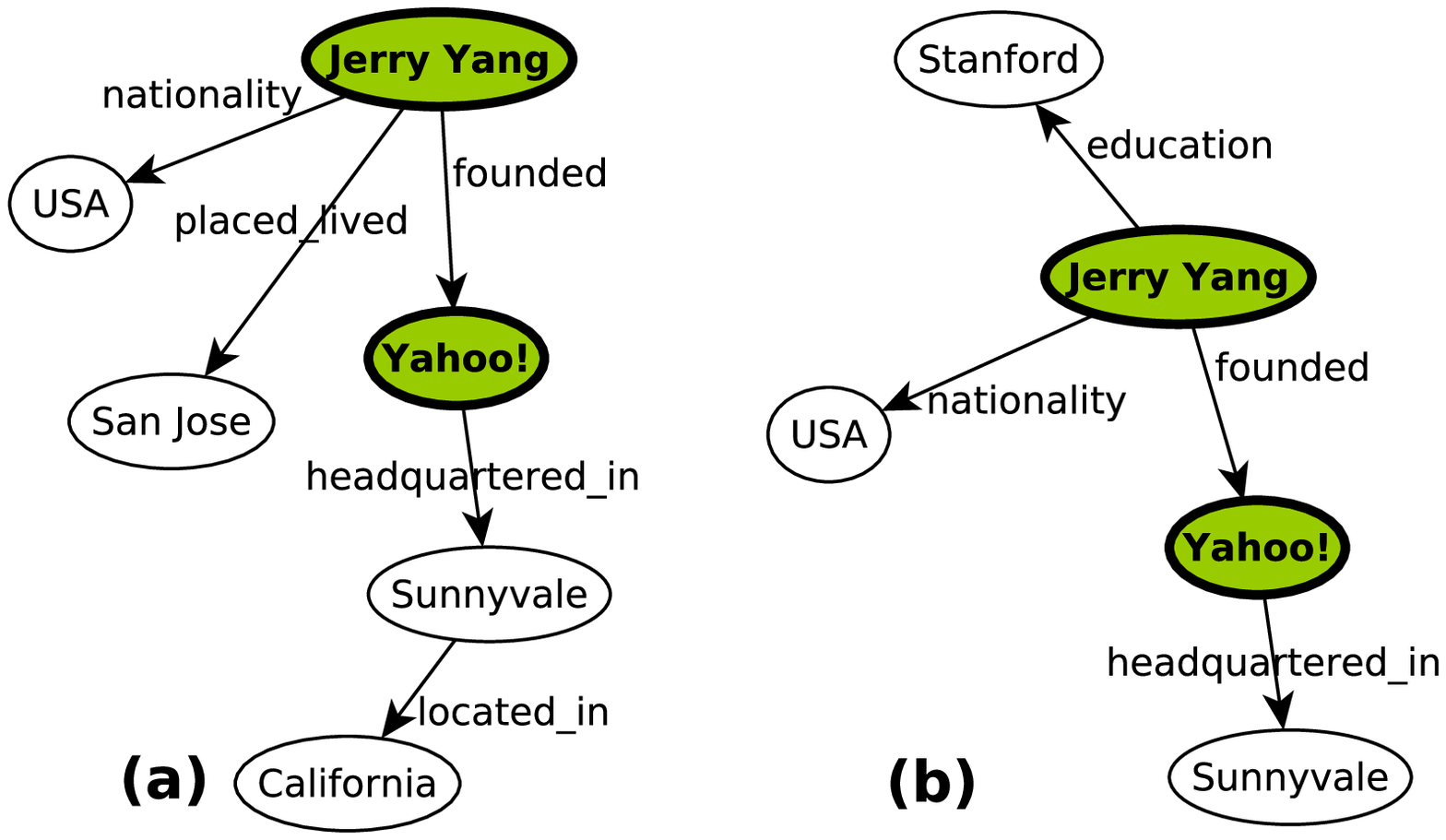}%\vspace{-3mm}
\caption{Two Query Graphs in Fig.\ref{fig:neighbor-graph}}
\label{fig:query-graph}
\end{minipage}%\vspace{-1mm}
\hspace{0.1cm}
\begin{minipage}[b]{0.33\linewidth}
\centering
  \includegraphics[width = 1.0\linewidth, keepaspectratio = true, scale=0.6]{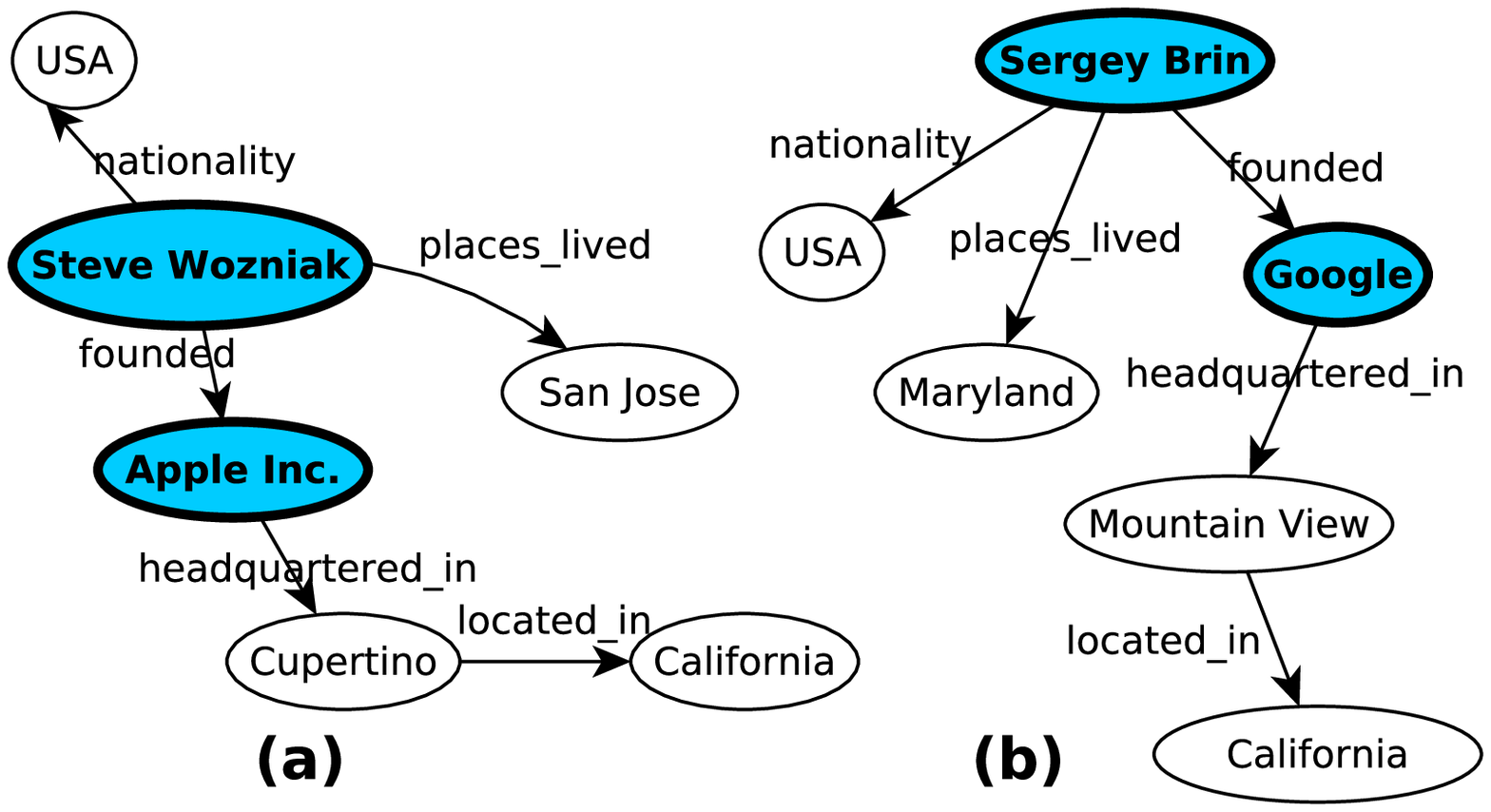}%\vspace{-3mm}
\caption{Two Answer Graphs for Fig.\ref{fig:query-graph}(a)}
\label{fig:answer-graph}
\end{minipage}%\vspace{-1mm}
\hspace{0.1cm}
\begin{minipage}[b]{0.33\linewidth}
\centering
  \includegraphics[width = 1.0\linewidth, keepaspectratio = true, scale=0.4]{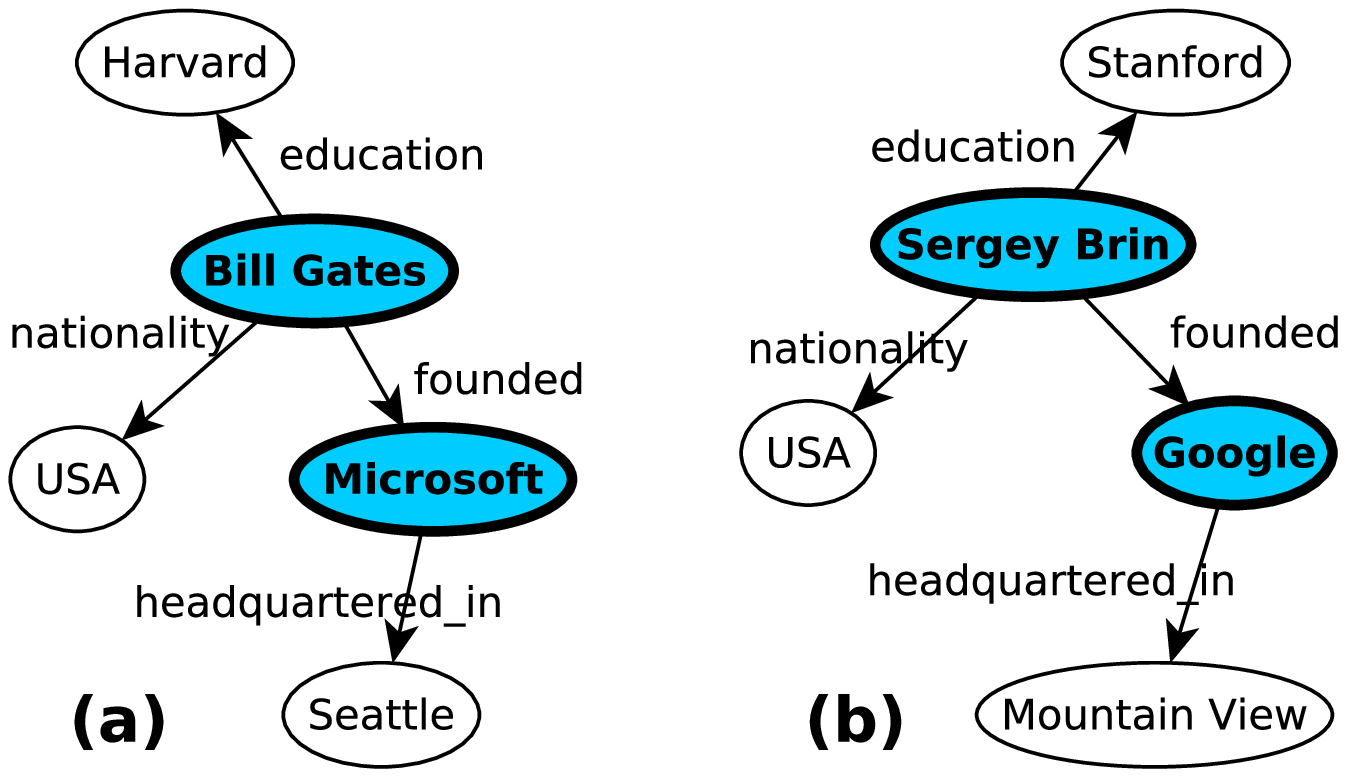}%\vspace{-3mm}
\caption{Two Answer Graphs for Fig.\ref{fig:query-graph}(b)}
\label{fig:answer-graph1}
\end{minipage}%\vspace{-1mm}
\end{figure*}

We define $\textsf{score}_t(t')$ by matching the inter-entity relationships
of $t$ and that of $t'$, which entails matching two graphs constructed
from $t$ and $t'$, respectively.  To this end, we define the
{\em neighborhood graph} for a tuple, which is based on the
concept of undirected path.
An \emph{undirected path} is a path whose edges are not necessarily
oriented in the same direction.  Unless otherwise stated, we will
refer to undirected path simply as ``path''.  We consider undirected
path because an edge incident on a node can represent an important
relationship with another node, regardless of its direction.
More formally, a path $p$ is a sequence of edges $e_1,\ldots,e_n$ and
we say each edge $e_i\in p$.  The path connects two nodes $v_0$ and $v_n$ through
intermediate nodes $v_1,\ldots,v_{n-1}$, where either
$e_i$=$(v_{i-1},v_i)$ or $e_i$=$(v_i,v_{i-1})$, for all
$1$$\leq$$i$$\leq$$n$.  %%%\footnote{Since a path is undirected,
%%%we will not make distinction between starting entity and ending entity.}
The length of the path, $len(p)$, is $n$ and the endpoints of the path,
$ends(p)$, are $\{v_0, v_n\}$.  Note that there is no undirected cycle
in a path, i.e., the entities $v_0,\ldots,v_n$ are all distinct.
%\vspace{-1mm}

\begin{definition}%%%[Neighborhood Graph]
\label{def:ng}
The \textbf{\em neighborhood graph} of query tuple $t$, denoted $H_t$, is the \emph{weakly connected
subgraph}\footnote{A directed graph is \emph{weakly connected}
if there exists an undirected path between every pair of vertices.}
of data graph $G$ that consists of all nodes reachable from at least one query entity by an
undirected path of $d$ or less edges (including query entities themselves)
and the edges on all such paths.
The \emph{path length threshold}, $d$, is an
input parameter.  More formally, the nodes and edges in $H_t$
are defined as follows:

$V(H_t) = \{v \lvert v \in V(G)$ and $\exists p$ s.t. $ends(p)$=$\{v_i,v\}$ where $v_i \in t \text{, } len(p) \leq d\}$;

$E(H_t) = \{e \lvert e \in E(G)$ and $\exists p$ s.t. $ends(p)$=$\{v_i,v\}$ where $v_i \in t, len(p) \leq d$, and $e \in p\}$.%\vspace{-1mm}
\end{definition}

\begin{example}[Neighborhood Graph]
Given the data graph in Fig.\ref{fig:example-graph},
Fig.\ref{fig:neighbor-graph} shows the neighborhood graph for query tuple
\etuple{Jerry Yang}{Yahoo!} with path length threshold $d$=2.
The nodes in dark color are the query entities.%\vspace{-1mm}
\end{example}

Intuitively, the neighborhood graph, by capturing how query entities and
other entities in their neighborhood are related to each other,
represents ``features'' of the query tuple that are to be matched in
query answers.  It can thus be viewed as a hidden query graph derived
for capturing user's query intent.
We are unlikely to find query answers that exactly match the
neighborhood graph.  It is however possible to find exact matches to its subgraphs.
Such subgraphs are all query graphs and
their exact matches are approximate answers that match the neighborhood graph
to different extents.%\vspace{-1mm}

\begin{definition}%%%[Query Graph]
\label{def:qgraph}
A \textbf{\em query graph} $Q$
is a weakly connected subgraph of $H_t$ that contains all the query entities.
We use $\mathcal{Q}_t$ to denote the set of all query graphs for $t$, i.e.,
$\mathcal{Q}_t$=\{$Q| Q$ is a weakly connected subgraph of $H_t$ s.t.
$\forall v \in t$, $v \in V(Q)$\}.%\vspace{-1mm}
\end{definition}

%%%\begin{example}[Query Graph]
Continuing the running example, Fig.\ref{fig:query-graph} shows two query graphs
for the neighborhood graph in Fig.\ref{fig:neighbor-graph}.
%%%\end{example}

Echoing the intuition behind neighborhood graph, the definitions of
answer graph and answer tuple are based on the idea that an answer
tuple is similar to the query tuple if their entities participate
in similar relationships in their neighborhoods.%\vspace{-1mm} %%%The definitions are as follows.

\begin{definition}%%%[Answer Graph and Answer Tuple]
\label{def:ansGraph}
An \textbf{\em answer graph} $A$ to a query graph $Q$
is a weakly connected subgraph of $G$ that is isomorphic to $Q$.
%%%, except that node identifiers may not necessarily match.
Formally, there exists a bijection $f$$:$$V(Q)$$\rightarrow$$V(A)$ such that:
\begin{list}{$\bullet$}
{ \setlength{\leftmargin}{1em} \setlength{\itemsep}{-1pt} }
\item For every edge $e = (v_i,v_j) \in E(Q)$, there exists an edge $e'=(f(v_i),f(v_j)) \in E(A)$ such that $label(e) = label(e')$;
\item For every edge $e' = (u_i,u_j) \in E(A)$, there exists $e = (f^{-1}(u_i),f^{-1}(u_j)) \in E(Q)$ such that $label(e) = label(e')$.
\end{list}%\vspace{-1mm}
For a query tuple $t$=$\langle v_1, \ldots, v_n \rangle$,
the \textbf{\em answer tuple} in $A$ is
$t_A$=$\langle f(v_1), \ldots, f(v_n)\rangle$.
We also call $t_A$ the \emph{projection} of $A$.

We use $\mathcal{A}_Q$ to denote the set of all answer graphs of $Q$.
We note that a query graph (tuple) trivially matches itself, therefore is not
considered an answer graph (tuple).%\vspace{-1mm}
\end{definition}

\begin{example}[Answer Graph and Answer Tuple]
Fig.\ref{fig:answer-graph} and Fig.\ref{fig:answer-graph1} each show two
answer graphs for query graphs Fig.\ref{fig:query-graph}(a) and
Fig.\ref{fig:query-graph}(b), respectively.
%%%They are exact matches to their query graphs, but approximate matches to the neighborhood graph.
%%%Note that matching nodes in query and answer graphs are not necessarily identical.
The answer tuples in
Fig.\ref{fig:answer-graph} are $\langle \entity{Steve Wozniak}$, $\entity{Apple Inc.} \rangle$ and \etuple{Sergey Brin}{Google}.
The answer tuples in Fig.\ref{fig:answer-graph1} are
\etuple{Bill Gates}{Microsoft} and \etuple{Sergey Brin}{Google}.%\vspace{-1mm}
\end{example}

\begin{definition}%%%[Answer Tuple Score]
\label{def:scoringAnsTuple}
The set of answer tuples for query tuple $t$ are
$\{ t_A | A$$\in$$\mathcal{A}_Q, Q$$\in$$\mathcal{Q}_t\}$.
The \emph{\textbf{score of an answer $t'$}} is given by%\vspace{-2mm}
\begin{align}
\label{eq:ranking_function}
\textsf{score}_t(t')= \max_{A \in \mathcal{A}_Q, Q\in \mathcal{Q}_t} \{{\textsf{score}_Q}(A) | t'=t_A \}%\vspace{-2mm}
\end{align}
The score of an answer graph $A$ (${\textsf{score}_{Q}}(A)$) captures
$A$'s similarity to query graph $Q$. Its equation is given in
Sec.\ref{sec:agraphscore}.%\vspace{-1mm}
\end{definition}

The same answer tuple $t'$ may be projected from multiple answer graphs,
which can match different query graphs.  For instance,
Figs.~\ref{fig:answer-graph}(b) and \ref{fig:answer-graph1}(b),
which are answers to different query graphs, have the same
projection---$\langle \entity{Sergey Brin}$, $\entity{Google}\rangle$.
By Eq.~(\ref{eq:ranking_function}), the highest score attained by the
answer graphs is assigned as the score of $t'$, capturing how well $t'$ matches $t$.

%A trivial similar case is when
%an answer graph is a subgraph of another answer graph.  (Thus their
%corresponding query graphs also form a subgraph-supergraph relationship.)
%They will project on to the same answer tuple.
%Furthermore, two answer graphs of the same query graph may project on
%to the same answer tuple too.
%Consider the simple query graph in Fig. \ref{fig:query-graph2}(a),
%for query tuple \etuple{Jerry Yang}{California}.  The two answer graphs
%in Figs.~\ref{fig:query-graph2}(b) and (c) correspond to the same
%answer tuple \etuple{Steve Wozniak}{California}.

%%%%Our objective is to return only distinct answer tuples.
%Given a query tuple $t$, the set of answer tuples are
%$\{ t' | t' \leftarrow A, A\in \mathcal{A}_Q, Q\in \mathcal{Q}_t\}$.

\section{Query Graph Discovery}
\label{sec:qgraph}

\subsection{Maximal Query Graph}\label{sec:mqg}

The concept of neighborhood graph $H_t$ (Def.\ref{def:ng}) was
formed to capture the features of a query tuple $t$ to be matched
by answer tuples.  Given a
well-connected large data graph, $H_t$ itself can be quite large, even
under a small path length threshold $d$.  For example, using Freebase
as the data graph, the query tuple \etuple{Jerry Yang}{Yahoo!} produces
a neighborhood graph with $800$K nodes and $900$K edges, for $d$=$2$.
Such a large $H_t$ makes query semantics obscure, because there might
be only few nodes and edges in it that capture important
relationships in the neighborhood of $t$.  %%%Moreover, in finding
%%approximate answers to a large $H_t$, the query evaluation cost is high.
\reminder{discuss $d$ in experiments.}
%%%$d$ was empirically chosen to be 2 since our experiments showed that a lesser $d$ failed to generate a neighborhood graph connecting all the query entities in some cases, while bigger values of $d$ resulted in large neighborhood graphs containing around $50\%$ of the edges in the well connected data graph.

\system{GQBE}'s query graph discovery component
constructs a weighted \emph{maximal query graph} (MQG)
from the neighborhood graph $H_t$.
MQG is expected to be drastically smaller than $H_t$
and capture only important features of the query tuple.
%It is termed the maximal query graph because it is treated as
%the largest query graph allowed for $t$ and any other query graph must
%be its subgraph.
We now define MQG and discuss its discovery algorithm.
%\vspace{-1mm}
%%%and finally we present a preprocessing idea for improving the query discovery process.
\reminder{CL: Provide an example of maximal query graph.}

\begin{definition}%%%%[Maximal Query Graph]
\label{def:mqg}
The \textbf{\em maximal query graph} $MQG_t$, given a parameter $m$, is a weakly connected
subgraph of the neighborhood graph $H_t$ that maximizes total edge weight $\sum_{e} \textsf{w}(e)$
while satisfying (1) it contains all query entities in $t$ and
(2) it has $m$ edges. The weight of an edge $e$ in $H_t$, $\textsf{w}(e)$,
is defined in Sec.\ref{sec:edgeweight}. %\vspace{-1mm}
%%%$MQG_t \preceq H'_t \text{ s.t } \forall v \in t, v \in V(MQG_t) \text{ and } \lvert E(MQG_t) \lvert \leq m$.
\end{definition}

There are two challenges in finding $MQG_t$ by directly going after the
above definition.  First, a weakly connected subgraph of $H_t$ with
exactly $m$ edges may not exist for an arbitrary $m$.
A trivial value of $m$ that
guarantees the existence of the corresponding $MQG_t$ is $|E(H_t)|$,
because $H_t$ itself is weakly connected.
This value could be too large, which is exactly why we aim to make
$MQG_t$ substantially smaller than $H_t$.
Second, even if $MQG_t$ exists for an $m$, finding it requires
maximizing the total edge weight, which is a hard problem as given in
Theorem~\ref{th:mqg}.%\vspace{-1mm}

\begin{theorem}
\label{th:mqg}
The decision version of finding the maximal query graph $MQG_t$ for
an $m$ is NP-hard.%\vspace{-1mm}
\end{theorem}
\begin{proof}
We prove the NP-hardness by reduction from the NP-hard constrained Steiner
network (CSN) problem~\cite{LZZC09}.
Given an undirected connected graph
$G_1$=$(V,E)$ with non-negative weight $w(e)$ for every $e$$\in$$E$, a
subset $V_n$$\subset$$V$, and a positive integer $m$, the CSN problem is
to find a connected subgraph $G'$=$(V',E')$ with the smallest total edge
weight, where $V_n$$\subseteq$$V'$ and $|E'|$=$m$.  The polynomial-time
reduction from the CSN problem to the MQG problem is by
transforming $G_1$ to $G_2$, where each edge $e$ is given an
arbitrary direction and a new weight $w'(e)$=$W$$-$$w(e)$, where
$W$=$\sum_{e \in E} w(e)$.  Let $V_n$ be the query tuple. The maximal
query graph $MQG_{V_n}$ found from $G_2$ provides a CSN in $G_1$, by
ignoring edge direction.  This completes the proof.
\end{proof}

\begin{algorithm}[t]
\label{alg:mqg}
\caption{Discovering the Maximal Query Graph}
\LinesNumbered
\footnotesize

\SetKw{KwAnd}{and}
\SetKw{KwDownTo}{downto}

\KwIn{neighborhood graph $H_t$, query tuple $t$, an integer $r$}

\KwOut{maximal query graph $MQG_t$}

\BlankLine
$m \leftarrow \frac{r}{\lvert t \lvert + 1}$; $V(MQG_t) \leftarrow \phi$; $E(MQG_t) \leftarrow \phi$; $\mathcal{G} \leftarrow \phi$;

\ForEach{$v_i \in t$}
{
    $G_{v_i} \leftarrow $ use DFS to obtain the subgraph containing vertices (and their incident edges) that connect to other $v_j$ in $t$ only through $v_i$\;
    $\mathcal{G} \leftarrow \mathcal{G} \cup \{G_{v_i}\}$\;
}\label{ln:beginloop}
$G_{core} \leftarrow $ use DFS to obtain the subgraph containing vertices and edges on undirected paths between query entities\;
$\mathcal{G} \leftarrow \mathcal{G} \cup \{G_{core}\}$\;\label{ln:findgraph}
\ForEach{$G \in \mathcal{G}$}
{
    $step \leftarrow 1$; $s_1 \leftarrow 0$; $s \leftarrow m$;

    \While{$s>0$}
    {
        $M_{s} \leftarrow$ the weakly connected component found from the top-$s$ edges of $G$ that contains all of $G$'s query entities;\label{ln:findms}

        \If{$M_s$ exists}
        {
          \lIf{$|E(M_{s})| = m$}{{\bf break}\label{ln:equalm}}
          
          \If{$|E(M_{s})| < m$}
          {
            $s_1 \leftarrow s$;

            \lIf{$step = -1$}{{\bf break}}
          }
          \If{$|E(M_{s})| > m$}
          {
            \If{$s_1 > 0$}
            {
                $s \leftarrow s_1$; {\bf break};
            }
            $s_2 \leftarrow s$; $step \leftarrow -1$;
          }
        }
        $s \leftarrow s+step$;
    }

    \lIf{$s = 0$}{$s \leftarrow s_2$}

	%%$s \leftarrow \min$\{$i\ \bigl \lvert \ \lvert E(M_{i}) \lvert = m$\}; //$M_{i}$ is the weakly connected component found from the top-$i$ edges of $G$\label{ln:equalm}\\
    %%\If{$s$ \emph{does not exist}}
    %%{
    %%	$s \leftarrow \max$\{$i\ \bigl \lvert \ \lvert E(M_{i}) \lvert < m$\}\;\label{ln:lessm}
     %%   \lIf{$s$ \emph{does not exist}}{$s \leftarrow \min$\{$i\ \bigl \lvert \ \lvert E(M_{i}) \lvert > m$\}\label{ln:greaterm}}
    %%}
    $V(MQG_t) \leftarrow V(MQG_t) \cup V(M_s)$\;
    $E(MQG_t) \leftarrow E(MQG_t) \cup E(M_s)$\;
}
\end{algorithm}

Based on the theoretical analysis, we present a greedy method (Alg.\ref{alg:mqg})
to find a plausible sub-optimal graph of edge cardinality \emph{close} to a given $m$.
%%%Unless otherwise noted, we will simply call the resulting $M_s$ the maximal
%%%query graph and denote it by $MQG_t$, to avoid excessive terminologies.
The value of $m$ is empirically chosen to be much smaller than $|E(H_t)|$.
Consider edges of $H_t$ in descending order of weight $\textsf{w}(e)$.
We use $G_s$ to denote the graph formed by the top $s$ edges with the largest weights, which itself
may not be weakly connected.  We use $M_s$ to denote the weakly connected
component of $G_s$ containing all query entities in $t$, if it exists.
Our method finds the smallest $s$ such that $|E(M_s)|$=$m$ (Line~\ref{ln:equalm}).  If such an $M_s$ does not
exist, the method chooses $s_1$, the largest $s$ such that $|E(M_s)|$$<$$m$.
If that still does not exist, it chooses $s_2$, the smallest $s$ such that $|E(M_s)|$$>$$m$,
whose existence is guaranteed because $|E(H_t)|$$>$$m$.
%%%The implementation of the method works by setting $s$ to $m$ initially and
%%%increasing $s$ by $1$ in each iteration until $|E(M_s)|$=$m$ (then it terminates)
%%%or $|E(M_s)|$$>$$m$.  If $|E(M_s)|$$>$$m$, it then decreases $s$ by $1$
%%%in each iteration until $|E(M_s)|$$<$$m$ or $s$=$0$.
For each $s$ value, the method employs a depth-first search (DFS) starting from
a query entity in $G_s$, if present, to check the existence of $M_s$ (Line~\ref{ln:findms}).

The $M_s$ found by this method may be unbalanced.
Query entities with more neighbors in $H_t$ likely have
more prominent representation in the resulting $M_s$.
A balanced graph should instead have a fair number of edges
associated with each query entity.   Therefore, we
further propose a divide-and-conquer mechanism to construct a balanced
$MQG_t$.  The idea is to break $H_t$ into $n$$+$$1$ weakly connected subgraphs.
One is the {\em core graph}, which includes all the $n$ query entities
in $t$ and all undirected paths between query entities.
Other $n$ subgraphs are for the $n$ query entities individually, where
the subgraph for entity $v_i$ includes all entities (and their incident
edges) that connect to other query entities only through $v_i$.
The subgraphs are identified by a DFS starting
from each query entity (Lines~\ref{ln:beginloop}-\ref{ln:findgraph} of Alg.\ref{alg:mqg}).  During the DFS from $v_i$, all edges
on the undirected paths reaching any other query entity within distance
$d$ belong to the core graph, and other edges belong to $v_i$'s individual subgraph.
The method then applies the aforementioned greedy algorithm to find $n$$+$$1$
weakly connected components, one for each subgraph, that contain the query entities
in corresponding subgraphs.  Since the core graph
connects all query entities, the $n$$+$$1$ components altogether form a weakly
connected subgraph of $H_t$, which becomes the final $MQG_t$.
For an empirically chosen small $r$ as the target size of $MQG_t$, we set
the target size for each individual component to be $\frac{r}{n+1}$, aiming at
a balanced $MQG_t$.

\spara{Complexity Analysis of Alg.\ref{alg:mqg}}\hspace{2mm}
In the aforementioned divide-and-conquer method,
if on average there are $r'$=$\frac{\lvert E(H_t) \lvert}{n+1}$ edges
in each subgraph, finding the subgraph by DFS and sorting its $r'$
edges takes $O(r'\log r')$ time.
Given the top-$s$ edges of a subgraph, checking if the weakly
connected component $M_s$ exists using DFS requires $O(s)$ time.
Suppose on average $c$ iterations are required to find the appropriate $s$.
Let $m$=$\frac{r}{n+1}$ be the average target edge cardinality of each
subgraph.  Since the method initializes $s$ with $m$, the largest value
$s$ can attain is $m$$+$$c$.  So the time for discovering
$M_s$ for each subgraph is $O(r' \log r'$$+$$c$$\times$$(m$$+$$c$)).
For all $n$$+$$1$ subgraphs, the total time required
to find the final $MQG_t$ is $O((n$$+$$1)\times(r' \log r'$$+$$c$$\times$$(m$$+$$c)))$.
For the queries used in our experiments on Freebase, given an empirically
chosen small $r$=$15$, $s$$\ll$$\lvert E(H_t)\lvert$ and on average $c$=$22$.
%%%In practice, computing the \emph{damping factor} of edges in a very large graph like $H_t$ might be expensive so we compute the \emph{damping factor} of only those edges in $MQG_t$ and prune if found to be unimportant.

\subsection{Edge Weighting}\label{sec:edgeweight}
%\vspace{-1mm}
The definition of $MQG_t$ (Def.\ref{def:mqg}) depends on edge
weights.  There can be various plausible weighting schemes.
We propose a weighting function based on several heuristic ideas.
The weight of an edge $e$ in $H_t$, $\textsf{w}(e)$,
is proportional to its inverse edge label frequency ($\textsf{ief}(e)$) and
inversely proportional to its participation degree ($\textsf{p}(e)$), given by%\vspace{-2mm}
\begin{align}
\label{eq:edge_wt_function}
\textsf{w}(e) = \textsf{ief}(e)\;/\;\textsf{p}(e)%\vspace{-2mm}
\end{align}

%\vspace{-1mm}
\spara{Inverse Edge Label Frequency}\hspace{2mm}
Edge labels that appear frequently in the entire
data graph $G$ are often less important. For example,
edges labeled \edge{founded} (for a company's founders) can be rare and more important than edges labeled
\edge{nationality} (for a person's nationality).
We capture this by the \emph{inverse edge label frequency}.%\vspace{-2mm}
\begin{align}
\label{eq:ief}
\textsf{ief}(e) = \log\;(|E(G)|\;/\;\#label(e))%\vspace{-2mm}
\end{align}
where $|E(G)|$ is the number of edges in $G$, and $\#label(e)$ is the
number of edges in $G$ with the same label as $e$.
%The inverse edge label frequency is analogous to the concept of inverse document frequency (IDF) in information retrieval. It captures the global importance of an edge. If an edge label appears frequently, it represents less unique information than what a rare edge label represents. For example, the relationship that entity \underline{Jerry Yang} is CEO of \underline{Yahoo!} might be more unique than the relationship that his nationality is American, which might in turn be more unique than the relationship that his gender is male. What sets apart \underline{Jerry Yang} from other American males is that he is also the CEO of a company. This brings about the intuition that the relationships which are less abundant might be more important than those that are very common. Thus the weight of an edge should be proportional to its inverse edge label frequency:
%\begin{align}
%w_e \varpropto ief_e
%\end{align}

\spara{Participation Degree}\hspace{2mm}
%The inverse edge frequency described above captures the global occurrence of an edge. We use the concept of \emph{participation} to measure the local occurrence of an edge w.r.t its two ends.
The {\em participation degree} $p(e)$ of an edge $e$=$(u,v)$ is the number of edges in $G$ that share the same label and one of $e$'s end nodes. Formally,%\vspace{-2mm}
\begin{align}
\label{eq:participation}
\hspace{-2mm} \textsf{p}(e)=\lvert\ \{e'\textrm{=}(u',v')\ |\ label(e)\textrm{=}label(e'), u'\textrm{=}u \vee v'\textrm{=}v\}\ \lvert%\vspace{-2mm}
\end{align}

While $\textsf{ief}(e)$ captures the global frequencies of edge
labels, $\textsf{p}(e)$ measures their local frequencies---an edge
is less important if there are other edges incident on the
same node with the same label.  For instance, \edge{employment}
might be a relatively rare edge globally but not necessarily
locally to a company.  Specifically, consider the edges representing
the \edge{employment} relationship between a company
and its \emph{many} employees and the edges for the \edge{board member}
relationship between the company and its \emph{few} board members.  The latter
edges are more significant.
%%%, because most likely a company has much more employees than board members.

%\begin{example}[Participation]
%The participation value of edge $e$, \emph{located\_in = (Sunnyvale, California)}, in Fig.\ref{fig:example-graph} is 6. This is because there are four other edges labeled \emph{located\_in} incident on node \emph{California} each contributing 1 towards $p(e)$. 2 is contribute towards $p(e)$ from $e$ itself since it matches both ends. The other four edges contribute 1 each since none of those have \emph{Sunnyvale} at the other end.
%\end{example}

Note that $\textsf{ief}(e)$ and $\textsf{p}(e)$ are precomputed
offline, since they are query-independent and only rely on the
data graph $G$.
\reminder{CL: do they depend on $G$ or reduced neighborhood graph? Can they be precomputed?}

\subsection{Preprocessing: Reduced Neighborhood Graph}
%\vspace{-1mm}

The discussion so far focuses on discovering $MQG_t$ from $H_t$.
The neighborhood graph $H_t$ may have clearly unimportant edges.
As a preprocessing step, \system{GQBE} removes such edges from $H_t$
before applying Alg.\ref{alg:mqg}.  The reduced size of $H_t$
not only makes the execution of Alg.\ref{alg:mqg} more efficient
but also helps prevent clearly unimportant edges from getting into $MQG_t$.

Consider the neighborhood graph $H_t$ in Fig.\ref{fig:neighbor-graph}, based
on the data graph excerpt in Fig.\ref{fig:example-graph}.
Edge $e_1$=\edgeends{Jerry Yang}{Stanford} and $label$($e_1$)=\edge{education}.
Two other edges labeled \edge{education}, $e_2$ and $e_3$, are also
incident on node \entity{Stanford}.  The neighborhood
graph from a complete real-world data graph may contain many such edges
for people graduated from Stanford University.
Among these edges, $e_1$ represents an important relationship between
\entity{Stanford} and query entity \entity{Jerry Yang}, while other edges represent
relationships between \entity{Stanford} and other entities, which are
deemed unimportant with respect to the query tuple.

We formalize the definition of \emph{unimportant edges} as follows.
Given an edge $e$=$(u,v) \in E(H_t)$, $e$ is unimportant if it is
unimportant from the perspective of its either end, $u$ or $v$, i.e.,
$\text{if } e \in UE(u)\ \text{or}\ e \in UE(v)$.
Given a node $v \in V(H_t)$, $E(v)$ denotes the edges
incident on $v$ in $H_t$.  $E(v)$ is partitioned into three disjoint subsets---the important edges $IE(v)$, the unimportant edges $UE(v)$ and the rest---defined as follows:
{\flushleft \vspace{-2mm}$IE(v)$=}

\vspace{-1mm}$\{e \in E(v)\ \lvert\ \exists v_i$$\in$$t, p \text{ s.t. } e$$\in$$p, ends(p)$=$\{v,v_i\},len(p)$$\leq$$d\}$;
{\flushleft \vspace{-1mm}$UE(v)$=}

\vspace{-1mm}$\{e \in E(v)\ \lvert\ e$$\notin$$IE(v), \exists e'$$\in$$IE(v) \text{ s.t. } label(e)$=$label(e')$,

$(e$=$(u,v) \wedge e'$=$(u',v)) \vee (e$=$(v,u) \wedge e'$=$(v,u'))\}$.\\
An edge $e$ incident on $v$ belongs to $IE(v)$ if
there exists a path between $v$ and any query entity in the query tuple $t$,
through $e$, with path length at most $d$.
For example, edge $e_1$ in Fig.\ref{fig:neighbor-graph} belongs to $IE(\entity{Stanford})$.
An edge $e$ belongs to $UE(v)$ if (1) it does not belong to $IE(v)$
(i.e., there exists no such aforementioned path) and (2) there exists
$e' \in IE(v)$ such that $e$ and $e'$ have the same label and they are
both either incoming into or outgoing from $v$.
%%%In Fig.\ref{fig:neighbor-graph}, $e_2$ and $e_3$ have the same label
%%%and direction as $e_1$.  Furthermore, there is no path between
%%%\entity{Stanford} and \entity{Jerry Yang} (or \entity{Yahoo}) that contains
%%%$e_2$ or $e_3$.  Hence $e_2$ and $e_3$ belong to $UE(v)$.
By this definition, $e_2$ and $e_3$ belong to $UE(v)$ in Fig.\ref{fig:neighbor-graph},
since $e_1$ belongs to $IE(v)$.  In the same
neighborhood graph, $e_4$ is in neither $IE(v)$ nor $UE(v)$.
%%%Although there does not exist a path between \entity{Jerry Yang} (or \entity{Yahoo}) and \entity{Stanford} through $e_4$, there are no important edges incoming into \entity{Stanford} with the same label \edge{founded}.

All edges deemed unimportant by the above definition are removed from $H_t$.
The resulting graph may not be weakly connected anymore and may have
multiple weakly connected components.~\footnote{A weakly connected component of a directed graph
is a maximal subgraph where an undirected path exists for every pair of vertices.}
Theorem~\ref{th:damping-factor} states that one of the components---called
the \emph{reduced neighborhood graph}, denoted $H'_t$---contains all query
entities in $t$.
In other words, $H'_t$ is the largest weakly connected subgraph of $H_t$
containing all query entities and no unimportant edges.
Alg.\ref{alg:mqg} is applied on $H'_t$ (instead of $H_t$)
to produce $MQG_t$.
Since the techniques in the ensuing discussion only operate on $MQG_t$,
the distinction between $H_t$ and $H'_t$ will not be further noted.%\vspace{-1mm}

\begin{theorem}
\label{th:damping-factor}
Given the neighborhood graph $H_t$ for a query tuple $t$, the reduced
neighborhood graph $H'_t$ always exists.%\vspace{-1mm}
\end{theorem}
\begin{proof} 
We prove by contradiction.  Suppose that, after removal of all unimportant
edges, $H_t$ becomes a disconnected graph, of which none of the weakly
connected components contains all the query entities.  The deletion of unimportant edges
must have disconnected at least a pair of query entities, say, $v_i$ and $v_j$.
By Def.~\ref{def:ng}, before removal of unimportant edges, $H_t$ must have at least 
a path $p$ of length at most $d$ between $v_i$ and $v_j$.  
By the definition of unimportant edges, every edge $e$=$(u,v)$ on $p$ belongs to 
both $IE(u)$ and $IE(v)$ and thus cannot be an unimportant edge. 
However, the fact that $v_i$ and $v_j$ become disconnected implies that $p$ consists of 
at least one unimportant edge which is deleted. 
This presents a contradiction and completes the proof. 
\end{proof}

\subsection{Multi-tuple Queries}
\label{sec:multituple}

The query graph discovery component of \system{GQBE} essentially derives
a user's query intent from input query tuples.  For that, a single query
tuple might not be sufficient.  While the experiment results in
Sec.\ref{sec:exp} show that a single-tuple query obtains excellent
accuracy in many cases, the results also exhibit that allowing multiple
query tuples often helps improve query answer accuracy.
This is because important relationships commonly associated with multiple
query tuples express the user intent more precisely. For instance,
suppose a user has provided two query tuples together---\etuple{Jerry Yang}{Yahoo!}
and \etuple{Steve Wozniak}{Apple Inc.}.  The query entities in both tuples
share common properties such as \edge{places\_lived} in \entity{San Jose}
and \edge{headquartered\_in} a city in \entity{California}, as shown in
Fig.\ref{fig:example-graph}.  This might indicate that the user is
interested in finding people from San Jose who founded technology companies
in California.

Given a set of tuples $T$, \system{GQBE} aims at finding top-$k$ answer
tuples similar to $T$ collectively.  To accomplish this, one approach is to
discover and evaluate the maximal query graphs (MQGs) of individual query
tuples.  The scores of a common answer tuple for multiple query tuples
can then be aggregated.  This has two potential drawbacks: (1) Our concern
of not being able to well capture user intent still remains.
If $k$ is not large enough, a good answer tuple may not appear in
enough individual top-$k$ answer lists, resulting in poor aggregated score.
(2) It can become expensive to evaluate multiple MQGs.

\begin{figure}[tb!]
\centering
  \includegraphics[width = 0.85\linewidth, keepaspectratio = true, scale=0.25]{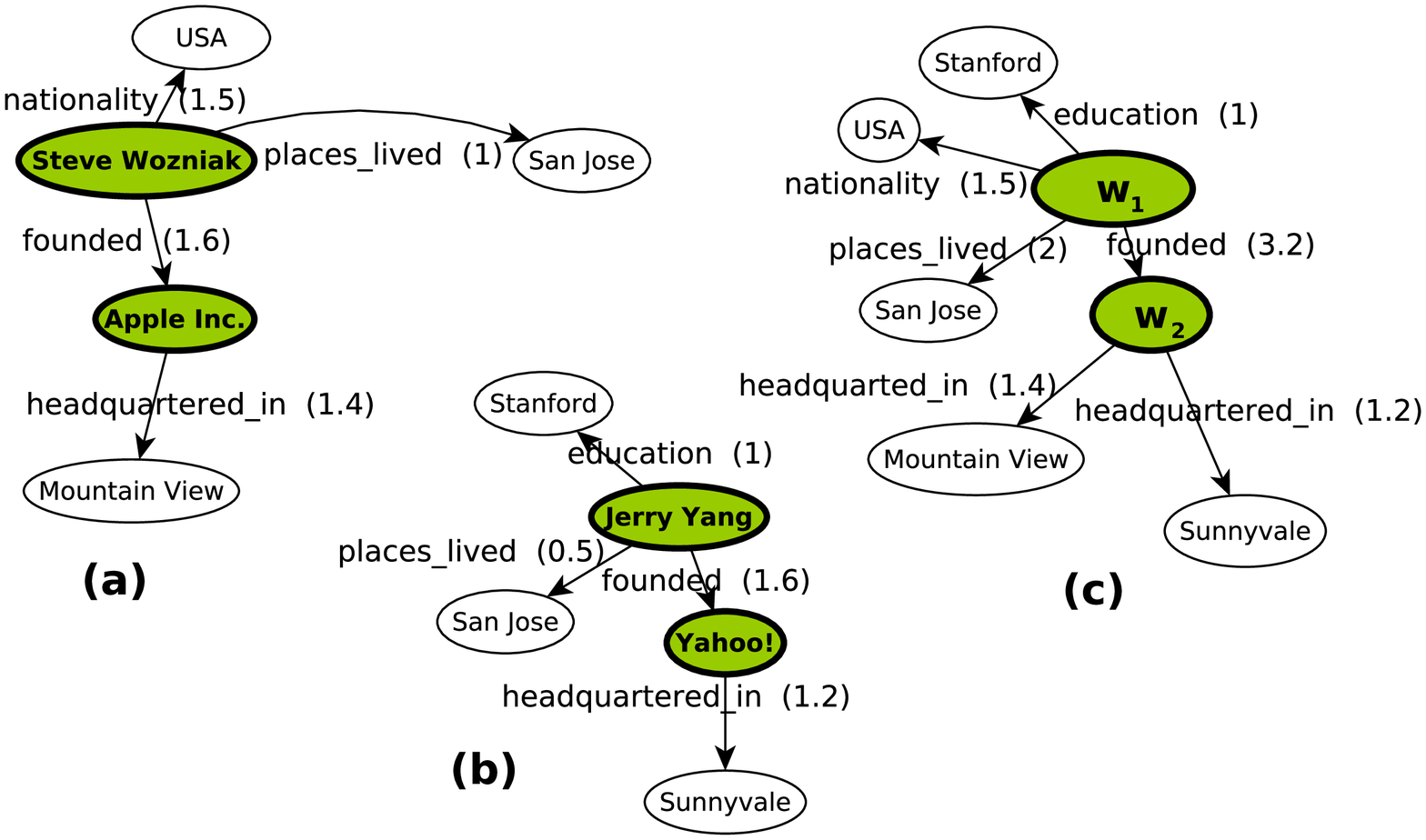}%\vspace{-4mm}
\caption{Merging Maximal Query Graphs} 
\label{fig:multiple-tuple}
\end{figure}

We approach this problem by producing a merged and re-weighted MQG that
captures the importance of edges with respect to their presence across
multiple MQGs. The merged MQG is then processed by the same method for single-tuple
queries.  \system{GQBE} employs a simple strategy to
merge multiple MQGs.  The individual MQG for a query tuple
$t_i$=$\langle v_1^i, v_2^i, \ldots, v_n^i \rangle$$\in$$T$
is denoted $M_{t_i}$.  A virtual MQG $M_{t_i}'$ is created for every
$M_{t_i}$ by replacing the query entities $v_1^i, v_2^i, \ldots, v_n^i$
in $M_{t_i}$ with corresponding virtual entities $w_1, w_2, \ldots, w_n$ in
$M_{t_i}'$.  Formally, there exists a bijective function $g$$:$$V(M_{t_i})$$\rightarrow$$V(M_{t_i}')$
such that (1) $g(v_j^i)$=$w_j$ and $g(v)$=$v$ if $v$$\notin$$t_i$, and
(2) $\forall e$=$(u,v)$$\in$$E(M_{t_i})$, there exists an edge $e'$=$(g(u),g(v))$ $\in$$E(M_{t_i}')$
such that $label(e)$=$label(e')$; $\forall e'$=$(u',v')$$\in$$E(M_{t_i}')$,
$\exists$$e$ =$(g^{-1}(u'),g^{-1}(v'))$$\in$$E(M_{t_i})$ such that $label(e)$=$label(e')$.

The merged MQG is denoted $MQG_T$.  It is produced by including vertices
and edges in all $M_{t_i}'$, merging identical virtual and regular vertices,
and merging identical edges that bear the same label and the same vertices
on both ends.  %%%Essentially, the multiple MQGs are overlaid on each other.
Formally,%\vspace{-1mm}
\begin{align}
\textstyle V(MQG_T)=\bigcup\limits_{t_i \in T} V(M_{t_i}') \text{ and } E(MQG_T)=\bigcup\limits_{t_i \in T} E(M_{t_i}').\nonumber%\vspace{-2mm}
\end{align}
The edge cardinality of $MQG_T$ might be larger than the target size $r$.
Thus Alg.\ref{alg:mqg} proposed in Sec.\ref{sec:mqg} is
also used to trim $MQG_T$ to a size close to $r$.
In $MQG_T$, the weight of an edge $e$ is given by $c*\textsf{w}_{max}(e)$,
where $c$ is the number of $M_{t_i}'$ containing $e$ and $\textsf{w}_{max}(e)$
is its maximal weight among all such $M_{t_i}'$.%\vspace{-1mm}

\begin{example}[Merging Maximal Query Graphs]
Let Figs.~\ref{fig:multiple-tuple} (a) and (b) be the $M_{t_i}$
for query tuples \etuple{Steve Wozniak}{Apple Inc.} and
\etuple{Jerry Yang}{Yahoo!}, respectively.
Fig.\ref{fig:multiple-tuple}(c) is the merged $MQG_T$.
Note that entities \entity{Steve Wozniak} and \entity{Jerry Yang} are
mapped to $w_1$ in their respective $M_{t_i}'$ (not shown,
for its mapping from $M_{t_i}$ is simple) and are merged into
\entity{$w_1$} in $MQG_T$.  Similarly, entities \entity{Apple Inc.} and
\entity{Yahoo!} are mapped and merged into $w_2$.
The two \edge{founded} edges, appearing in both individual $M_{t_i}$ and
sharing identical vertices on both ends ($w_1$ and $w_2$) in the
corresponding $M_{t_i}'$, are merged in $MQG_T$.  Similarly the two
\edge{places\_lived} edges are merged. However, the two
\edge{headquartered\_in} edges are not merged, since they share only one
end ($w_2$) in $M_{t_i}'$.  The edges \edge{nationality} and
\edge{education}, which appear in only one $M_{t_i}$, are also present in
$MQG_T$.  The number next to each edge is its weight.%\vspace{-1mm}
\end{example}

In comparison to evaluating a single-tuple query, the extra overhead
in handling a multi-tuple query includes creating multiple MQGs, which
is $\lvert T \lvert$ times the average cost of discovering an individual
MQG, and merging them, which is linear in the total edge cardinality of all MQGs.

\section{Answer Space Modeling}
\label{sec:modeling}

Given the maximal query graph $MQG_t$ for a tuple $t$, we model the
space of possible query graphs by a lattice.  We further discuss the
scoring of answer graphs by how well they match query graphs.

%\vspace{-1mm}
\subsection{Query Lattice}
\begin{definition}\label{def:lattice}%%%[Query Lattice]
The \textbf{\em query lattice} $\lattice$ is a partially ordered set (poset)
($\mathcal{QG}_t$, $\prec$), where $\prec$ represents the
subgraph-supergraph subsumption relation and $\mathcal{QG}_t$ is the subset of query
graphs (Def.\ref{def:qgraph}) that are subgraphs of $MQG_t$, i.e.,
$\mathcal{QG}_t$$=$$\{Q|Q\in \mathcal{Q}_t \text{ and } Q \preceq MQG_t\}$.
The top element (root) of the poset is thus $MQG_t$.  When represented by a
Hasse diagram, the poset is a directed acyclic graph, in which each node
corresponds to a distinct query graph in $\mathcal{QG}_t$.  Thus we shall use the terms
{\em lattice node} and {\em query graph} interchangeably.
The {\em children} ({\em parents}) of a lattice node $Q$ are its subgraphs
(supergraphs) with one less (more) edge, as defined below.
%\vspace{-2mm}
\begin{align}
&\textsf{Children}(Q) =  \{Q' | Q' \in \mathcal{QG}_t, Q' \prec Q, |E(Q)|\!-\!|E(Q')|\!=\!1 \}& \nonumber\\%\vspace{-2mm}\\
&\textsf{Parents}(Q)  =  \{Q' | Q' \in \mathcal{QG}_t, Q \prec Q', |E(Q')|\!-\!|E(Q)|\!=\!1\}& \nonumber%\vspace{-1mm}
\end{align}
\end{definition}

The leaf nodes of $\lattice$ constitute of the \emph{minimal query trees}, which
are those query graphs that cannot be made any simpler and yet still keep all the
query entities connected.%\vspace{-1mm}

\begin{definition}%%%[Minimal Query Tree]
A query graph $Q$ is a \textbf{\em minimal query tree} if none of its subgraphs is also a query graph.
In other words, removing any edge from $Q$ will disqualify it from being a query graph---the resulting graph
either is not weakly connected or does not contain all the query entities. Note that such a $Q$ must be
a tree.%\vspace{-1mm}
\end{definition}

\begin{figure}[tb!]
\centering
  \includegraphics[width = 1.0\linewidth, keepaspectratio = true, scale=0.4]{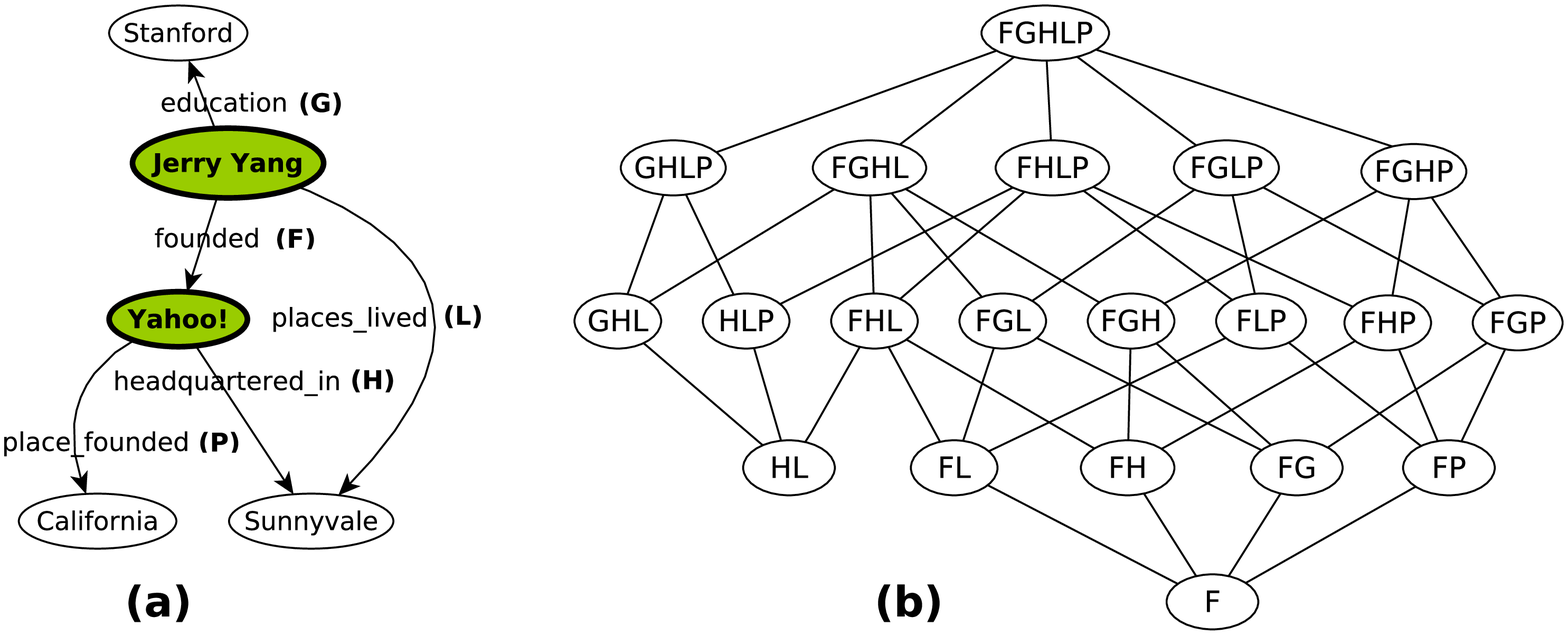}%\vspace{-2mm}
\caption{Maximal Query Graph and Query Lattice}
\label{fig:lattice}%\vspace{-1mm}
\end{figure}

%\vspace{-1mm}
\begin{example}[Query Lattice and Minimal Query Tree]
Fig.\ref{fig:lattice}(a) shows a maximal query graph $MQG_t$, which
contains two query entities in shaded circles and five edges $F,G,H,L,$
and $P$.  Its corresponding query lattice $\lattice$ is in
Fig.\ref{fig:lattice}(b).  The root node of $\lattice$, denoted $FGHLP$,
represents $MQG_t$ itself.  The two bottom nodes, $F$ and
$HL$, are the two minimal query trees.  Each lattice node is a distinct
subgraph of $MQG_t$.  For example, the node $FLP$ represents a query graph
with only edges $F, L$ and $P$.  Note that there is no lattice node for
$GLP$, which is not a valid query graph since it is not connected.%\vspace{-1mm}
\end{example}

The construction of the query lattice, i.e., the generation of query graphs
corresponding to its nodes, is integrated with its exploration.  In other
words, the lattice is built in a ``lazy'' manner---a lattice node is not
generated until the query algorithm (Sec.\ref{sec:processing}) must
evaluate it.  The lattice nodes are generated in a bottom-up way.  A node
is generated by adding exactly one appropriate edge to the query graph
for one of its children.  The generation of bottom nodes, i.e., the minimal query trees,
is described below.

By definition, a minimal query tree can only contain edges on undirected paths
between query entities.  Hence, it must be a subgraph of the weakly connected
component $M_s$ found from the core graph described in Sec.\ref{sec:mqg}.
To generate all minimal query trees, our method enumerates all distinct spanning
trees of $M_s$ by the technique in~\cite{GabowM78} and then trim them.
Specifically, given one such spanning tree, all non-query entities (nodes) of degree one
along with their edges are deleted.  The deletion is performed iteratively until there is no
such node. The result is a minimal query tree.  Only distinct minimal query trees are kept.
Enumerating all spanning trees in a large graph is expensive.
%%% since its time complexity is a function of the number of spanning trees.
However, in our experiments on the Freebase dataset, the $MQG_t$ discovered by the
approach in Sec.\ref{sec:qgraph} mostly contains less than $15$ edges.
Hence, the $M_s$ from the core graph is also empirically small, for which the cost
of enumerating all spanning trees is negligible. %%% Specifically, it typically has only $4$ edges.
%%%the number of spanning trees in it around 3. On an average, there are 2 minimal query trees for each MQG.
%Enumerating all spanning trees using ~\cite{GabowM78} takes $O(\lvert{E(MQG_t)}\lvert + \lvert{V(MQG_t)}\lvert + S_N*\lvert{E(MQG_t)}\lvert)$ time where $S_N$ is the number of spanning trees in $MQG_t$. Trimming each spanning tree to a minimal query tree takes $O(\lvert V(MQG_t) \lvert)$ time.

\subsection{Answer Graph Scoring Function}\label{sec:agraphscore}
%\vspace{-1mm}
The score of an answer graph $A$ (${\textsf{score}_{Q}}(A)$) captures
$A$'s similarity to the query graph $Q$. It is defined below and
is to be plugged into Eq.~(\ref{eq:ranking_function}) for defining
answer tuple score.%\vspace{-1mm}
\begin{align}
\begin{aligned}
\label{eq:scoreFunction}
{\textsf{score}_{Q}}(A) & = {\textsf{s\_score}}(Q) + {\textsf{c\_score}_{Q}}(A)\\%\vspace{-2mm}
{\textsf{s\_score}}(Q) & = \;\;\;\;\; \sum_{e \in E(Q)}{\textsf{w}(e)}\\%\vspace{-2mm}
{\textsf{c\_score}_{Q}}(A) & = \!\!\!\!\!\! \sum_{\substack{ e=(u,v) \in E(Q) \\ e'=(f(u), f(v)) \in E(A)}}\!\!\!\!\!\!\!\!\!\!\!\!\textsf{match}(e, e')%\vspace{-2mm}
\end{aligned}
\end{align}%\vspace{-2mm}

In Eq.~(\ref{eq:scoreFunction}), ${\textsf{score}_{Q}}(A)$ sums up
two components---the \emph{structure score} of $Q$ (${\textsf{s\_score}}(Q)$)
and the \emph{content score} for $A$ matching $Q$ (${\textsf{c\_score}_{Q}}(A)$).
${\textsf{s\_score}}(Q)$ is the total edge weight of $Q$.
It measures the important structure in $MQG_t$ that is captured by $Q$ and thus by $A$.
${\textsf{c\_score}_{Q}}(A)$ is the total extra credit for identical nodes among
the matching nodes in $A$ and $Q$ given by $f$---the bijection between
$V(Q)$ and $V(A)$ as in Def.\ref{def:ansGraph}.
For instance, among the $6$ pairs of matching nodes between
Fig.\ref{fig:query-graph}(a) and Fig.\ref{fig:answer-graph}(a),
the identical matching nodes are \entity{USA}, \entity{San Jose}
and \entity{California}. The rationale for the extra credit is that although
node matching is not mandatory, the more nodes are matched, the more
similar $A$ and $Q$ are.

The extra credit is defined by the following function $\textsf{match}(e,e')$.
Note that it does not award an identical matching node excessively.
Instead, only a fraction of $\textsf{w}(e)$ is awarded, where the
denominator is either $|E(u)|$ or $|E(v)|$. ($E(u)$ are the edges
incident on $u$ in $MQG_t$.) This heuristic is based on that,
when $u$ and $f(u)$ are identical, many of their neighbors can be also
identical matching nodes.
%\vspace{-2mm}
\begin{multline}
\label{eq:match}
	\textsf{match}(e,e')\text{=}
	\begin{cases}
	\frac{\textsf{w}(e)}{\lvert{E(u)}\lvert} & \text{if } u\text{=}f(u)\\
	\frac{\textsf{w}(e)}{\lvert{E(v)}\lvert} & \text{if } v\text{=}f(v)\\
	\frac{\textsf{w}(e)}{min(\lvert{E(u)}\lvert,\lvert{E(v)}\lvert)} & \text{if } u\text{=}f(u), v\text{=}f(v)\\
	0 & \text{otherwise}
	\end{cases}%\vspace{-2mm}
\end{multline}

In discovering $MQG_t$ from $H_t$ by Alg.\ref{alg:mqg}, the weights
of edges in $H_t$ are defined by Eq.~(\ref{eq:edge_wt_function}) which
does not consider an edge's distance from the query tuple.
The rationale behind the design is to obtain a balanced $MQG_t$ which
includes not only edges incident on query entities but also those
in the larger neighborhood.  For scoring answers by
Eq.~(\ref{eq:scoreFunction}) and Eq.~(\ref{eq:match}), however,
our empirical observations show it is imperative to differentiate the
importance of edges in $MQG_t$ with respect to query entities,
in order to capture how well an answer graph matches $MQG_t$.
Edges closer to query entities convey more meaningful relationships
than those farther away. Hence, we define edge depth ($\textsf{d}(e)$)
as follows. The larger $\textsf{d}(e)$ is, the less important $e$ is.

\spara{Edge Depth}\hspace{2mm}
The depth $\textsf{d}(e)$ of an edge $e$=$(u,v)$ is its smallest
distance to any query entity $v_i \in t$,  \ie%\vspace{-2mm}
\begin{align}
\textsf{d}(e) = \min_{v_i \in t} \min_{u,v} \{\textsf{dist}(u,v_i),\textsf{dist}(v,v_i)\}%\vspace{-3mm}
\end{align}
Here, $\textsf{dist}(.,.)$ is the shortest length of all undirected paths
in $MQG_t$ between the two nodes.

In summary, \system{GQBE} uses Eq.~(\ref{eq:edge_wt_function}) as the definition
of $\textsf{w}(e)$ in weighting edges in $H_t$.  After $MQG_t$ is discovered from $H_t$
by Alg.\ref{alg:mqg}, it uses the following Eq.~(\ref{eq:mqg_weighting_function}) as
the definition of $\textsf{w}(e)$ in weighting edges in $MQG_t$.
Eq.~(\ref{eq:mqg_weighting_function}) incorporates $\textsf{d}(e)$ into Eq.~(\ref{eq:edge_wt_function}).
The answer graph scoring functions Eq.~(\ref{eq:scoreFunction}) and Eq.~(\ref{eq:match}) are based
on Eq.~(\ref{eq:mqg_weighting_function}).
%\vspace{-2mm}
\begin{align}
\label{eq:mqg_weighting_function}
\textsf{w}(e) = \textsf{ief}(e)\;/\;(\textsf{p}(e) \times \textsf{d}^2(e))
\end{align}

\section{Query Processing}
\label{sec:processing}

The query processing component of \system{GQBE} takes the maximal query
graph $MQG_t$ (Sec.\ref{sec:qgraph}) and the query lattice $\lattice$
(Sec.\ref{sec:modeling}) and finds answer graphs matching the query
graphs in $\lattice$. Before we discuss how $\lattice$ is
evaluated (Sec.\ref{sec:bestfirst}), we introduce the storage model and query plan for
processing one query graph (Sec.\ref{sec:onequery}).

%\vspace{-1mm}
\subsection{Processing One Query Graph}\label{sec:onequery}
The abstract data model of knowledge graph can be represented
by the Resource Description Framework (RDF)---the standard Semantic Web
data model.  In RDF, a data graph is parsed into a set of triples, each
representing an edge $e$=$(u,v)$.  A triple has the form
(subject, property, object), corresponding to ($u, label(e), v$).
Among different schemes of RDF data management, one important approach
is to use relational database techniques to store and query RDF graphs.
To store a data graph, we adopt this approach and, particularly, the
vertical partitioning method~\cite{abadi07}.
This method partitions a data graph into multiple
two-column tables.  Each table is for a distinct edge label and stores
all edges bearing that label.  The two columns are ($subj, obj$), for
the edges' source and destination nodes, respectively. For efficient query
processing, two in-memory search structures (specifically, hash tables)
are created on the table, using $subj$ and $obj$ as the hash keys, respectively.
The whole data graph is hashed in memory by this way, before any query comes in.
%%%~\cite{Sakr-2010}

Given the above storage scheme, to evaluate a query graph is to process a
multi-way join query.
For instance, the query graph in Fig.\ref{fig:lattice}(a) corresponds to
{\footnotesize \textsf{SELECT F.subj, F.obj FROM F,G,H,L,P WHERE F.subj=G.sbj
AND F.obj=H.subj AND F.subj=L.subj AND F.obj=P.subj AND H.obj=L.obj.}}
We use right-deep hash-joins to process such a query.
Consider the topmost join operator in a join tree for
query graph $Q$.  Its left operand is the \emph{build relation}
which is one of the two in-memory hash tables for an edge $e$.
Its right operand is the \emph{probe relation} which is a
hash table for another edge or a join subtree for
$Q'$=$Q$$-$$e$ (i.e., the resulting graph of removing $e$ from $Q$).
For instance, one possible join tree for the aforementioned
query is $G$$\bowtie$$(F$$\bowtie$$(P$$\bowtie$$(H$$\bowtie$$L)))$.
With regard to its topmost join operator, the left operand
is $G$'s hash table that uses $G.sbj$ as the hash key, and the right
operand is $(F$$\bowtie$$(P$$\bowtie$$(H$$\bowtie$$L)))$.
The hash-join operator iterates through tuples from the
probe relation, finds matching tuples from the build relation,
and joins them to form answer tuples.

%%\vspace{-1mm}
\subsection{Best-first Exploration of Query Lattice}\label{sec:bestfirst}

Given a query lattice, a brute-force approach is to evaluate all lattice
nodes (query graphs) to find all answer tuples.  Its exhaustive nature
leads to clear inefficiency, since we only seek \topk\ answers.
Moreover, the potentially many queries are evaluated separately, without
sharing of computation.  Suppose query graph $Q$ is evaluated by the
aforementioned hash-join between the build relation for $e$ and the
probe relation for $Q'$.  By definition, $Q'$ is also a query graph in
the lattice, if $Q'$ is weakly connected and contains all query entities.
In other words, in processing $Q$, we would have processed one of its
children query graph $Q'$ in the lattice.

We propose Alg.\ref{alg:hierarchical}, which allows sharing of computation.
It explores the query lattice in a \emph{bottom-up} way, starting with the
minimal query trees, i.e., the bottom nodes.  After a query graph is processed,
its answers are materialized in files.
To process a query $Q$, at least one of its children $Q'$=$Q$$-$$e$ must have been
processed.  The materialized results for $Q'$ form the probe relation and
a hash table on $e$ is the build relation.

While any topological order would work for the bottom-up exploration,
Alg.\ref{alg:hierarchical} employs a {\em best-first} strategy
that always chooses to evaluate the most promising lattice node
$Q_{best}$ from a set of candidate nodes.  The gist is to process the
lattice nodes in the order of their upper-bound scores and $Q_{best}$
is the candidate with the highest upper-bound score (Line~\ref{ln:qbest}).
If processing $Q_{best}$ does not yield any answer graph, $Q_{best}$
and all its ancestors are pruned (Line~\ref{ln:prune}) and the
upper-bound scores of other candidate nodes are recalculated
(Line~\ref{ln:recalc}).  The algorithm terminates, without fully evaluating
all lattice nodes, when it has obtained
at least \emph{k} answer tuples with scores higher than the highest
possible upper-bound score among all unevaluated nodes
(Line~\ref{ln:terminate}).

For an arbitrary query graph $Q$, its upper-bound score is given by the best
possible score $Q$'s answer graphs can attain.
Deriving such upper-bound score based on ${\textsf{score}_{Q}}(A)$ in Eq.~(\ref{eq:scoreFunction}) leads to loose upper-bound.
${\textsf{score}_{Q}}(A)$ sums up the structure score of $Q$ (${\textsf{s\_score}}(Q)$)
and the content score for $A$ matching $Q$ (${\textsf{c\_score}_{Q}}(A)$).
While ${\textsf{s\_score}}(Q)$ only depends on $Q$ itself, ${\textsf{c\_score}_{Q}}(A)$ captures
the matching nodes in $A$ and $Q$. Without evaluating $Q$ to get $A$,
we can only assume perfect $\textsf{match}(e, e')$ in Eq.~(\ref{eq:scoreFunction}),
which is clearly an over-optimism.
Under such a loose upper-bound, it can be difficult to achieve an early termination of
lattice evaluation.

To alleviate this problem, \system{GQBE} takes a two-stage approach.
Its query algorithm first finds the top-$k'$ answers ($k'$$>$$k$) based on the structure score
${\textsf{s\_score}}(Q)$ only, i.e., the algorithm uses a simplified answer graph
scoring function ${\textsf{score}_{Q}}(A)={\textsf{s\_score}}(Q)$.
In the second stage, \system{GQBE} re-ranks the top-$k'$ answers by the
full scoring function Eq.~(\ref{eq:scoreFunction}) and returns the \topk\ answer tuples based
on the new scores.
Our experiments showed the best accuracy for $k$ ranging from $10$ to $25$
when $k'$ was set to around $100$. Lesser values of $k'$ lowered the accuracy and
higher values increased the running time of the algorithm.  In the ensuing discussion,
we will not further distinct $k'$ and $k$.
%%%Therefore, we also say ``\topk'' instead of ``top-$k'$''.
%%%So $k'$ is empirically chosen to be 100.

Below we provide the algorithm details.

%\vspace{-1mm}
\subsection{Details of the Best-first Exploration Algorithm}\label{sec:alg}
%\vspace{-1mm}
\spara{\emph{(1) Selecting ${\bf {Q}_{best}}$}} 

At any given moment during query lattice evaluation, the lattice nodes belong to
three mutually-exclusive sets---the evaluated, the unevaluated and the pruned.
A subset of the unevaluated nodes, denoted the \emph{lower-frontier} ($\lf$), are
candidates for the node to be evaluated next.
%%%Going by the bottom-up lattice evaluation described above, if the node to be evaluated
%%%is not a minimal query tree, then at least one of its children must be evaluated.
At the beginning, $\lf$ contains only the minimal query trees (Line~\ref{ln:init} of Alg.\ref{alg:hierarchical}).  After a node is
evaluated, all its parents are added to $\lf$ (Line~\ref{ln:insert}).  Therefore,
the nodes in $\lf$ either are minimal query trees or have at least one evaluated child:%\vspace{-2mm}
\begin{center}
$\lf = \{Q |\ Q\ \text{is not pruned}, \textsf{Children}(Q)$$=$$\emptyset \text{ or }$\\
$(\exists Q' \in \textsf{Children}(Q)\ \text{s.t.}\ Q'\ \text{is evaluated})\}$.
\end{center}

To choose $Q_{best}$ from $\lf$, the algorithm exploits two important properties,
dictated by the query lattice's structure.%\vspace{-1mm}
\begin{property}
\label{prop:latticeAnswerGraphs}
If $Q_1 \prec Q_2$, then $\forall A_2 \in \mathcal{A}_{Q_2}$, $\exists A_1 \in \mathcal{A}_{Q_1}$ s.t. $A_1 \prec A_2$ and $t_{A_1}$=$t_{A_2}$.%\vspace{-2mm}
\begin{proof}
If there exists an answer graph $A_2$ for a query graph $Q_2$, and there exists another query graph $Q_1$ that is a subgraph of $Q_2$, then there is a subgraph of $A_2$ that corresponds to $Q_1$. By Definition~\ref{def:ansGraph}, that corresponding subgraph of $A_2$ is an answer graph to $Q_1$. Since the two answer graphs share a subsumption relationship, the projections of the two yield the same answer tuple.
\end{proof}
\end{property}

Property~\ref{prop:latticeAnswerGraphs} says, if an answer tuple
$t_{A_2}$ is projected from answer graph $A_2$ to lattice node $Q_2$,
then every descendent of $Q_2$ must have at least one answer graph
subsumed by $A_2$ that projects to the same answer tuple.
Putting it in an informal way, an answer tuple
(graph) to a lattice node can always be ``grown'' from its descendant nodes
and thus ultimately from the minimal query trees.%\vspace{-1mm}
%%%Therefore, if evaluating $Q_1$ returns answer graph $A_1$ and corresponding $t_{A_1}$,
%%%the same answer tuple may be found in the ancestor nodes of $Q_1$.%\vspace{-1mm}

\begin{property}
\label{prop:scoringLatticeTuple}
If $Q_1$$\prec$$Q_2$, then $\textsf{s\_score}(Q_1)$$<$$\textsf{s\_score}(Q_2)$.%\vspace{-1mm}
\begin{proof}
If $Q_1 \prec Q_2$, then $Q_2$ contains all edges in $Q_1$ and at least one more. Thus the property holds by the definition of $\textsf{s\_score}(Q)$ in Eq.~(\ref{eq:scoreFunction}).
\end{proof}
\end{property}

Property~\ref{prop:scoringLatticeTuple} says that, if a lattice node $Q_2$ is an
ancestor of $Q_1$, $Q_2$ has a higher structure score.  This can be directly
proved by referring to the definition of  $\textsf{s\_score}(Q)$ in Eq.~(\ref{eq:scoreFunction}).

For each unevaluated candidate node $Q$ in $\lf$, we define an {\em upper-bound score},
which is the best score $Q$'s answer tuples can possibly attain.
The chosen node, $Q_{best}$, must have the highest upper-bound score among all the nodes in $\lf$.
By the two properties, if evaluating $Q$ returns an answer graph $A$,
$A$ has the potential to grow into an answer graph $A'$ to an ancestor node
$Q'$, i.e., $Q$$\prec$$Q'$ and $A$$\prec$$A'$.  In such a case,
$A$ and $A'$ are projected to the same answer tuple $t_{A}$=$t_{A'}$.
The answer tuple always gets the better score from $A'$, under the
simplified answer scoring function
${\textsf{score}_{Q}}(A)={\textsf{s\_score}}(Q)$, which Alg.\ref{alg:hierarchical}
adopts as mentioned in Sec.~\ref{sec:bestfirst}.
Hence, $Q$'s upper-bound score depends on its \emph{upper boundary}---
$Q$'s unpruned ancestors that have no unpruned parents.%\vspace{-1mm}
\begin{definition}%%%[Upper Boundary]
\label{def:ub}
The \textbf{\em upper boundary} of a node $Q$ in $\lf$, denoted $\ub(Q)$,
consists of nodes $Q'$ in the \textbf{\em upper-frontier} ($\uf$) that
subsume or equal to $Q$:%\vspace{-2mm}
\begin{align}
\ub(Q) = \{Q' |\ Q' \supgraph Q, Q'\in \uf \}, \text{where} \nonumber %\vspace{-2mm}
\end{align}
$\uf$ is the set of unpruned nodes without unpruned parents:
$\uf$$=$$\{Q |\ Q\ \text{is not pruned}, \nexists Q'\succ Q\ \text{s.t.}\ Q'\ \text{is not pruned}\}$.
\end{definition}
%%%\footnote{If multiple nodes in $\lf$ have the same upper-bound score, the node with the most edges is selected as $Q_{best}$. If multiple nodes have the same upper-bound score and edge count, the node with the highest query graph score ${\textsf{score}}(Q)$ is chosen. If multiple nodes have the same value for each of the above three criteria, one of them is randomly selected.}

\begin{definition}%%%[Upper-bound Score and Lower-bound Score]
\label{def:ubs}
The \textbf{\em upper-bound score} of a node $Q$ is the maximum score
of any query graph in its upper boundary:%\vspace{-2mm}
\begin{align}
\label{eq:ub}
\ubs(Q) = \operatorname*{max}_{Q'\in \ub(Q)}\textsf{s\_score}(Q')%\vspace{-2mm}
\end{align}
\end{definition}

\begin{algorithm}[t]
\label{alg:hierarchical}
\caption{\textsf{Best-first Exploration of Query Lattice}}
\LinesNumbered
\footnotesize

\SetKw{KwAnd}{and}
\SetKw{KwDownTo}{downto}

\KwIn{query lattice $\lattice$, query tuple $t$, and an integer $k$}

\KwOut{top-$k$ answer tuples}

\BlankLine
lower frontier $\lf$ $\leftarrow$ leaf nodes of $\lattice$; {\em Terminate} $\leftarrow$ {\bf false}\;\label{ln:init}

\While{{\bf not} Terminate}
{
    $Q_{best} \leftarrow$ node with the highest upper-bound score in $\lf$\;\label{ln:qbest}
    $\mathcal{A}_{Q_{best}} \leftarrow$ evaluate $Q_{best}$; (Sec.\ref{sec:onequery})\\
    \uIf{$\mathcal{A}_{Q_{best}}$=$\emptyset$}
    {
        prune $Q_{best}$ and all its ancestors from $\lattice$\;\label{ln:prune}
        recompute upper-bound scores of nodes in $\lf$ (Alg.~\ref{alg:recompute})\;\label{ln:recalc}
    }
    \Else{insert $\textsf{Parents}(Q_{best})$ into $\lf$\;\label{ln:insert}}
    \lIf{{\em top-$k$ answer tuples found [Theorem~\ref{th:terminate}]}}
    {
        {\em Terminate}$\leftarrow${\bf true}\label{ln:terminate}
    }
}
\end{algorithm}

\begin{figure}[tb!]
\centering
  \includegraphics[width = 1.05\linewidth, keepaspectratio = true, scale=0.4]{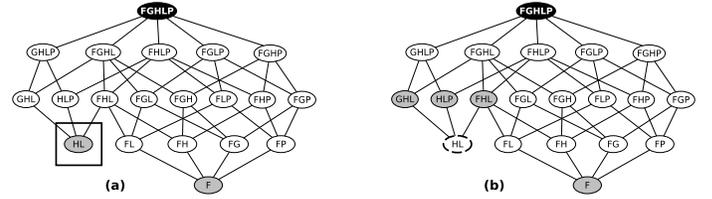}%\vspace{-3mm}
\caption{Evaluating Lattice in Figure~\ref{fig:lattice}(b)}
\label{fig:lattice-traversal}%\vspace{-2mm}
\end{figure}

\begin{example}[Lattice Evaluation]
Consider the lattice in Fig.\ref{fig:lattice-traversal}(a) where the lightly shaded nodes belong to the $\lf$ and the darkly shaded node belongs to $\uf$. At the beginning, only the minimal query trees belong to the $\lf$ and the maximal query graph belongs to the $\uf$. If \emph{HL} is chosen as $Q_{best}$ and evaluating it results in matching answer graphs, all its parents (\emph{GHL}, \emph{HLP} and \emph{FHL}) are added to $\lf$ as shown in Fig.\ref{fig:lattice-traversal}(b). The evaluated node \emph{HL} is represented in bold dashed node.
\end{example}

\spara{\emph{(2) Pruning and Lattice Recomputation}} 

A lattice node that does not have any answer graph is referred to as a {\em null node}.
If the most promising node $Q_{best}$ turns out to be a null node after evaluation, all its
ancestors are also null nodes based on Property~\ref{prop:up_close} below which follows
directly from Property~\ref{prop:latticeAnswerGraphs}. %\vspace{-1mm}
\begin{property}[Upward Closure]
\label{prop:up_close}
If $\mathcal{A}_{Q_1}=\emptyset$, then $\forall Q_2 \succ Q_1$, $\mathcal{A}_{Q_2}=\emptyset$.%\vspace{-1mm}
\begin{proof}
Suppose there is a query node $Q_2$ such that $Q_1 \prec Q_2$ and $\mathcal{A}_{Q_1}=\emptyset$, while $\mathcal{A}_{Q_2} \neq \emptyset$. By Property~\ref{prop:latticeAnswerGraphs}, for every answer graph $A$ in $\mathcal{A}_{Q_2}$, there must exist a subgraph of $A$ that belongs to $\mathcal{A}_{Q_1}$. This is a contradiction and completes the proof.
\end{proof}
\end{property}

\begin{algorithm}[t]
\label{alg:recompute}
\caption{\textsf{Recomputing Upper-bound Scores}}
\LinesNumbered
\footnotesize

\SetKw{KwAnd}{and}
\SetKw{KwDownTo}{downto}

\KwIn{query lattice $\lattice$, null node $Q_{best}$, and lower-frontier $\lf$}

\KwOut{$U(Q)$ for all $Q$ in $\lf$}

\BlankLine
\ForEach{$Q \in \lf$}
{
	$\nb \leftarrow \phi$; // set of new upper boundary candidates of $Q$.\\
    \ForEach{$Q' \in \ub(Q) \cap \ub(Q_{best})$}
    {
        $\ub(Q) \leftarrow \ub(Q)\setminus\{Q'\}$;\\
        $\uf \leftarrow \uf\setminus\{Q'\}$;\\
%        $E(Q'') \leftarrow E(Q)$;\\
 %       $E(Q'') \leftarrow E(Q'')\cup (E(Q')\setminus E(Q_{best}))$;\\
  %      $i \leftarrow 0$;\\
  		$V(Q'') \leftarrow V(Q')$;\\
  %		$E(Q'') \leftarrow E(Q')$;\\
        \ForEach{$e \in E(Q_{best})\setminus E(Q)$}
        {
   %         $i \leftarrow i+1$;\\
    %        $E(Q'') \leftarrow E(Q'')\cup ((E(Q_{best})\setminus E(Q)) \setminus e)$;\\
    			$E(Q'') \leftarrow E(Q') \setminus \{e\}$;\\
            find $Q_{sub}$, the weakly-connected component of $Q''$, containing all query entities;\\
            $\nb \leftarrow \nb \cup \{Q_{sub}\}$;\\
        }
    }
    \ForEach{$Q_{sub} \in \nb$}
    {
    		\If{$Q_{sub} \nprec $ (\text{any node in $\uf$ or $\nb$})}
    		{
    			$\ub(Q) \leftarrow \ub(Q) \cup\{Q_{sub}\}$, $\uf \leftarrow \uf\cup\{Q_{sub}\}$;
    		}
    }
    recompute $U(Q)$ using Eq.~(\ref{eq:ub});
}
\end{algorithm}

%%%\subsection{Recomputation of Upper-frontier and Upper-bound Scores} \label{sec:recompute}

\begin{figure}[tb!]
\centering
  \includegraphics[width = 1.05\linewidth, keepaspectratio = true, scale=0.4]{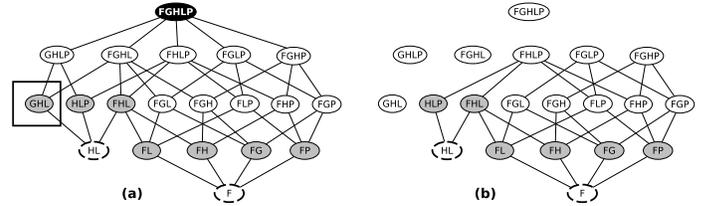}%\vspace{-1mm}
\caption{Recomputing Upper Boundary of Dirty Node $FG$}
\label{fig:recompute-lattice}%\vspace{-1mm}
\end{figure}

Based on Property~\ref{prop:up_close},
when $Q_{best}$ is evaluated to be a null node, Alg.\ref{alg:hierarchical} prunes
$Q_{best}$ and its ancestors, which changes the upper-frontier $\uf$.
It is worth noting that $Q_{best}$ itself may be an upper-frontier
node, in which case only $Q_{best}$ is pruned. In general, due to the
evaluation and pruning of nodes, $\lf$ and $\uf$ might overlap.
For nodes in $\lf$ that have at least one upper boundary
node among the pruned ones, the change of $\uf$ leads to changes
in their upper boundaries and, sometimes, their upper-bound scores too.
We refer to such nodes as {\em dirty nodes}. The rest of this section
presents an efficient method (Alg.~\ref{alg:recompute}) to recompute
the upper boundaries, and if changed, the upper-bound scores of the
dirty nodes.

Consider all the pairs $\langle Q, Q' \rangle$ such that
$Q$ is a dirty node in $\lf$, and $Q'$ is one of its pruned upper boundary nodes.
Three necessary conditions for a new candidate upper boundary node of $Q$ are that it is
%properties of a potential new upper boundary node for $Q$ are that it is
(1) a supergraph of $Q$, (2) a subgraph of $Q'$ and (3) not a supergraph of $Q_{best}$.
The subsumption relationships among these graphs can be visualized in a Venn diagram, as shown in Fig.\ref{fig:venn}.
If there are $q$ edges in $Q_{best}$ but not in $Q$ (the non-intersecting region of $Q_{best}$ in Fig.\ref{fig:venn}),
we create a set of $q$ distinct graphs $Q''$. Each $Q''$ contains all edges in $Q'$ except
exactly one of the aforementioned $q$ edges (Line 8 in Alg.~\ref{alg:recompute}).
For each $Q''$, we find $Q_{sub}$ which is the weakly connected component of $Q''$ containing all the query entities (Lines 9-10).
\begin{figure}[t]
\centering
  \includegraphics[width = 0.5\linewidth, keepaspectratio = true, scale=0.65]{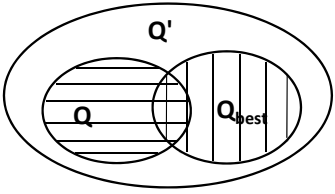}%\vspace{-2mm}
 \caption{Venn Diagram of Edges}
\label{fig:venn}
\end{figure}
Lemma~\ref{lemma:qsub_notnull} and \ref{lemma:exist_qsub} show that $Q_{sub}$ must be one of the unevaluated nodes after pruning
the ancestor nodes of $Q_{best}$ from $\lattice$.

\begin{lemma}
\label{lemma:qsub_notnull}
$Q_{sub}$ is a {\em query graph} and it does not belong to the pruned nodes of lattice $\lattice$.%\vspace{-1mm}
\end{lemma}
\begin{proof}
$Q_{sub}$ is a query graph because it is weakly connected and it contains all the input entities.
Suppose $Q_{sub}$ is a newly generated candidate upper boundary node from pair $\langle Q,Q' \rangle$ and $Q_{sub}$ belongs to the pruned nodes of lattice $\lattice$. This can happen only due to one of
the two reasons: 1) it is a supergraph of the current null node $Q_{best}$ or 2) it is an already pruned
node. The former cannot happen since the construction mechanism of $Q_{sub}$ proposed ensures that it is not a supergraph
of $Q_{best}$. the latter implies that $Q_{sub}$ was the supergraph of an previously evaluated null node (or $Q_{sub}$ itself was a null node). In this case, since
$Q_{sub} \prec Q'$, $Q'$ would also have been pruned and thus could not have been part of the upper-boundary. Hence $\langle Q,Q' \rangle$ cannot be a valid pair for recomputing the upper boundary if $Q_{sub}$ is a pruned node. This completes the proof.
\end{proof}
\begin{lemma}
\label{lemma:exist_qsub}
$Q \preceq Q_{sub}$.%\vspace{-1mm}
\end{lemma}
\begin{proof}
Based on Alg.~\ref{alg:recompute} described above, $Q''$ is the result of deleting one edge from $Q'$ and that edge does not belong to $Q$.
Therefore, $Q$ is subsumed by $Q''$.  By the same algorithm,  $Q_{sub}$ is the weakly connected component of $Q''$ that contains all the query entities. 
Since $Q$ already is weakly connected and contains all the query entities, $Q_{sub}$ must be a supergraph of $Q$. 
\end{proof}
If $Q_{sub}$ (a candidate new upper boundary node of $Q$) is not subsumed by any node in the upper-froniter or other candidate nodes, we add $Q_{sub}$ to $\ub$($Q$) and $\uf$ (Lines 11-13).
Finally, we recompute $Q$'s upper-bound score (Line 14). %%%The above method is performed for all pairs of $\langle Q, Q' \rangle$ (Line 1) to re-compute the new upper-frontier.
Theorem~\ref{th:correctness} justifies the correctness of the above procedure.%\vspace{-2mm}
\begin{theorem}
\label{th:correctness}
If $Q_{best}$ is evaluated to be a null node, then Alg.\ref{alg:recompute} identifies all new upper boundary nodes for every dirty node $Q$.
%The aforementioned procedure identifies all new upper boundary nodes for every dirty node $Q$.%\vspace{-1mm}
\end{theorem}
\begin{proof}
For any dirty node $Q$, its original upper boundary $\ub(Q)$ consists of two sets of nodes: (1) nodes that are not supergraphs of $Q_{best}$ and thus remain in the lattice, (2) nodes that are supergraphs of $Q_{best}$ and thus are pruned. By the definition of upper boundary node, no upper boundary node of $Q$ can be a subgraph of any node in set (1). So any new upper boundary node of $Q$ must be a subgraph of a node $Q'$ in set (2). For every pruned upper boundary node $Q'$ in set (2), the algorithm enumerates all (specifically $q$) possible children of $Q'$ that are not supergraphs of $Q_{best}$ but are supergraphs of $Q$.  For each enumerated graph $Q''$, the algorithm finds $Q_{sub}$, which is the weakly connected component of $Q''$ containing all query entities.  Thus all new upper boundary nodes of $Q$ are identified.
\end{proof}
%%%\begin{proof}
%%%Proof omitted due to space limitations.
%For $Q''$ to be a new upper boundary node of some $Q$ in the frontier pool, $Q''$ must satisfy
%the following necessary and sufficient conditions. (i) $Q''$ is not a supergraph of $Q_{best}$,
%which is a null node. (ii) $Q''$ is a supergraph of $Q$. (iii) $Q''$ is a subgraph of some
%$Q'$, which becomes a deleted node. (iv) $Q''$ must not be a subgraph of any other upper boundary
%nodes of $Q$. Now, observe that our method generates all $Q''$ satisfying (i), (ii), (iii), and
%also having the highest number of edges. Hence, all new upper boundary nodes are identified
%using the aforementioned technique.
%%%\end{proof}
%\vspace{-1mm}

\begin{example}[Recomputing Upper Boundary]
Consider the lattice in Fig.\ref{fig:recompute-lattice}(a) where nodes \emph{HL} and \emph{F} are the evaluated nodes and the lightly shaded nodes belong to the new $\lf$. If node \emph{GHL} is the currently evaluated null node $Q_{best}$ and \emph{FGHLP} is $Q'$, let \emph{FG} be the dirty node $Q$ whose upper boundary is to be recomputed.
The edges in $Q_{best}$ that are not present in $Q$ are \emph{H} and \emph{L}.
A new upper boundary node $Q''$ contains all edges in $Q'$ excepting exactly either \emph{H} or \emph{L}. This leads to two new upper boundary nodes, \emph{FGHP} and \emph{FGLP}, by removing \emph{L} and \emph{H} from \emph{FGHLP}, respectively.
Since \emph{FGHP} and \emph{FGLP} do not subsume each other and are not subgraphs of any other upper-frontier node, they are now part of $\ub$($Q$) and the new $\uf$.
Fig.\ref{fig:recompute-lattice}(b) shows the modified lattice where the pruned nodes are disconnected. \emph{FHLP} is another node in $\uf$ that is discovered
using dirty nodes such as \emph{FL} and \emph{HLP}.%\vspace{-2mm}
\end{example}

\spara{Complexity Analysis of Alg.\ref{alg:recompute}}\hspace{2mm}
The query graphs corresponding to lattice nodes are represented using bit vectors since we exactly know the edges involved in all the query graphs. The bit corresponding to an edge is set if its present in the query graph. Identifying the dirty nodes, null upper boundary nodes and building a new potential upper boundary node using a pair of nodes $\langle Q, Q' \rangle$, can be accomplished using bit operations and each step incurs $O(\lvert{E(MQG_t)}\lvert)$ time. Finding the weakly connected component of a potential upper boundary using DFS takes $O(\lvert{E(Q')}\lvert)$ time. If $\lattice_n$ is the set of all null nodes encountered in the lattice and there are $D_p$ such pairs for every null node and $q$ is the average number of potential new upper boundary nodes created per pair, the worst case time complexity of recomputing the upper-frontier is $O(\lvert \lattice_n \lvert \times D_p \times q \times \lvert{E(MQG_t)}\lvert)$. Our experimental results show low average values of $\lvert \lattice_n \lvert$, $D_p$ and $q$ with $\lvert \lattice_n \lvert$ being only 1\% of $\lvert \lattice \lvert$, $D_p$ around 8 and $q$ around 9. In practice, our upper-frontier recomputation algorithm quickly computes the dynamically changing lattice.

\spara{\emph{(3) Termination}} 

After $Q_{best}$ is evaluated, its answer tuples are $\{t_A | A$$\in$$\mathcal{A}_{Q_{best}}\}$.
For a $t_A$ projected from answer graph $A$, the score assigned by
$Q_{best}$ to $A$ (and thus $t_A$) is $\textsf{s\_score}(Q_{best})$, based on
${\textsf{score}_{Q}}(A)$$=$${\textsf{s\_score}}(Q)$---the simplified scoring function adopted by Alg.\ref{alg:hierarchical}.
If $t_A$ was also projected from already evaluated nodes, it has a current score.
By Def.\ref{def:scoringAnsTuple}, the final score of $t_A$ will be from its
best answer graph.  Hence, if $\textsf{s\_score}(Q_{best})$ is higher than
its current score, then its score is updated.
In this way, all found answer tuples so far are kept and their current scores
are maintained to be the highest scores they have received.
The algorithm terminates when the current score of the $k^{th}$ best answer tuple so far
is greater than the upper-bound score of the next $Q_{best}$ chosen by the algorithm,
by Theorem~\ref{th:terminate}.%\vspace{-1mm}
%%%The correctness of this termination criterion is based on Theorem~\ref{th:terminate}.%\vspace{-2mm}

\begin{theorem}\label{th:terminate}
If the score of the current $k^{th}$ best answer tuple is greater than $\ubs(Q_{best})$, then terminating the lattice evaluation guarantees that the current
top-$k$ answer tuples have scores higher than $\textsf{s\_score}(Q)$ for any unevaluated
query graph $Q$.
\end{theorem}
\begin{proof}
Suppose, upon termination, there is an unevaluated query graph $Q$ whose $\textsf{s\_score}(Q)$ is greater than the score of the $k^{th}$ answer tuple.
This implies that there exists some node in the lower-frontier whose upper-bound score is at least $\textsf{s\_score}(Q)$ and is thus greater than the score of the $k^{th}$ answer tuple. Since the termination condition precludes this, it is a contradiction to the initial assumption. We thus cannot have any unevaluated query graph whose structure score is greater than the $k^{th}$ answer tuple's score upon termination. 
\end{proof}
%%%\begin{proof}
%%%Proof omitted due to space limitations.
%Suppose there is a $i^{th}$ answer tuple $t'_i$ which has a score greater than $k^{th}$ answer tuple $t'_k$ ($i > k$).
%Then, this implies that there is some node in the lattice which has an \emph{upper bound} at least as much as the
%score of $t'_i$, which is greater than the score of $t'_k$. This implies that the \emph{upper bound} of the next node to traverse is atleast as much as the score of $t'_i$. This is a contradiction to the termination condition
%which specifies that the \emph{upper bound} of the next node to traverse in the lattice is lesser than $t'_k$. Hence, there can be no such $t'_i$ after the termination condition is satisfied.
%%%\end{proof}
%%\vspace{-1mm}
\reminder{Arijit: I added back the earlier proof, though I think the first line of the proof is somewhat unclear.
What do i-th and k-th answer tuples mean? -- based on ranking? also the following line is unclear: ``the termination condition
which specifies that the largest \emph{upper bound} in the lattice is lesser than $t'_k$''. -- what is largest upper
bound in the lattice?}

\iffalse
%\vspace{-1mm}
\spara{(4) Complexity Analysis}\hspace{2mm}
Joins are used to evaluate the lattice nodes. Minimal query trees might require multiple joins and other lattice nodes require a single join each.
In evaluating the latter, if on average,
the number of answer graphs for a lattice node is $j$, the time to evaluate a node by joining the answers of its child node and the new edge added to form the node is $O(j)$.
If $\lvert{\lattice_{e}}\lvert$ is the actual number of lattice nodes evaluated, the worst case scenario of query processing is $O(\lvert{\lattice_{e}}\lvert \times j)$. In practice, due to the pruning power of the best-first exploration technique, $\lvert{\lattice_{e}}\lvert \ll \lvert{\lattice}\lvert$. For the queries used in our experiments on Freebase, on average only 8\% of $\lvert{\lattice}\lvert$ is evaluated. The average number of answers to a lattice node, $j$, is 6500. Thus, the time to evaluate a single lattice node has a significant role in the total query processing time. Therefore, the query processing time is not only dependent on the size of $MQG_t$, but also on the join cardinality involving the edges.
\fi

\section{Experiments}\label{sec:exp}
This section presents our experiment results on the accuracy and
efficiency of \system{GQBE}.  The experiments were
conducted on a double quad-core $24$ GB Memory $2.0$ GHz Xeon server.

\begin{table} [t]
\centering
\scriptsize
\begin{tabular}{|p{6mm}|l|l|}
  \hline
  {\bf Query}  &   {\bf Query Tuple} &   {\bf Table Size} \\  \hline\hline
  F$_1$ & $\langle$Donald Knuth, Stanford University, Turing Award$\rangle$ & 18 \\
  F$_2$ & $\langle$Ford Motor, Lincoln, Lincoln MKS$\rangle$ & 25 \\
  F$_3$ & $\langle$Nike, Tiger Woods$\rangle$ & 20 \\
  F$_4$ & $\langle$Michael Phelps, Sportsman of the Year$\rangle$ & 55 \\
  F$_5$ & $\langle$Gautam Buddha, Buddhism$\rangle$ & 621 \\
  F$_6$ & $\langle$Manchester United, Malcolm Glazer$\rangle$ & 40 \\
  F$_7$ & $\langle$Boeing, Boeing C-22$\rangle$ & 89 \\
  F$_8$ & $\langle$David Beckham, A. C. Milan$\rangle$ & 94 \\
  F$_9$ & $\langle$Beijing, 2008 Summer Olympics$\rangle$ & 41 \\
  F$_{10}$ & $\langle$Microsoft, Microsoft Office$\rangle$ & 200 \\
  F$_{11}$ & $\langle$Jack Kirby, Ironman$\rangle$ & 25 \\
  F$_{12}$ & $\langle$Apple Inc, Sequoia Capital$\rangle$ & 300 \\
  F$_{13}$ & $\langle$Beethoven, Symphony No. 5$\rangle$ & 600 \\
  F$_{14}$ & $\langle$Uranium, Uranium-238$\rangle$ & 26 \\
  F$_{15}$ & $\langle$Microsoft Office, C++$\rangle$ & 300 \\
  F$_{16}$ & $\langle$Dennis Ritchie, C$\rangle$ & 163 \\
  F$_{17}$ & $\langle$Steven Spielberg, Minority Report$\rangle$ & 40 \\
  F$_{18}$ & $\langle$Jerry Yang, Yahoo!$\rangle$ & 8349 \\
  F$_{19}$ & $\langle$C$\rangle$ & 1240 \\
  F$_{20}$ & $\langle$TomKat$\rangle$ & 16 \\
  D$_1$ & $\langle$Alan Turing, Computer Scientist$\rangle$ & 52 \\
  D$_2$ & $\langle$David Beckham, Manchester United$\rangle$ & 273 \\
  D$_3$ & $\langle$Microsoft, Microsoft Excel$\rangle$ & 300 \\
  D$_4$ & $\langle$Steven Spielberg, Catch Me If You Can$\rangle$ & 37 \\
  D$_5$ & $\langle$Boeing C-40 Clipper, Boeing$\rangle$ & 118 \\
  D$_6$ & $\langle$Arnold Palmer, Sportsman of the year$\rangle$ & 251 \\
  D$_7$ & $\langle$Manchester City FC, Mansour bin Zayed Al Nahyan$\rangle$ & 40 \\
  D$_8$ & $\langle$Bjarne Stroustrup, C++$\rangle$ & 964 \\
  \hline
\end{tabular}
\vspace{1mm}
\caption{Queries and Ground Truth Table Size}
\label{tab:groundTruth}
\vspace{-2mm}
\end{table}

\vspace{-1mm}
{\flushleft \textbf{Datasets}}\hspace{2mm} The experiments were conducted over two
large real-world knowledge graphs, the Freebase~\cite{Bollacker+08freebase}
and DBpedia~\cite{AuerBK+07} datasets.
%Freebase is a collaborative knowledge base created by incorporating various
%sources including Wikipedia and taking user-generated tables, while DBpedia mainly contains structured information extracted from Wikipedia infoboxes.
%%%We used the Freebase and DBpedia data dumps from September 2011 and August 2011 respectively.
We preprocessed the graphs so that the kept nodes are all named
entities (\eg \entity{Stanford University}) and abstract concepts (\eg
\entity{Jewish people}). The resulting Freebase graph contains $28$M nodes,
$47$M edges, and $5,428$ distinct edge labels. The DBpedia graph contains
$759$K nodes, $2.6$M edges and $9,110$ distinct edge labels.
\reminder{say how we removed symmetric edges, data cleaning steps}

\vspace{-1mm}
{\flushleft \textbf{Methods Compared}}\hspace{2mm}
\system{GQBE} was compared with a \system{Baseline} and \system{NESS}~\cite{ness}.
\system{NESS} is a graph querying framework that finds approximate matches of query
graphs with unlabeled nodes which correspond to query entity nodes in MQG.
Note that, like other systems, \system{NESS} must take a query graph (instead of a query tuple)
as input.  Hence, we feed the MQG discovered by \system{GQBE} as the query
graph to \system{NESS}. For each node $v$ in the query graph, a set of
candidate nodes in the data graph are identified. Since, \system{NESS} does not consider
edge-labeled graphs, we adapted it by requiring each candidate node $v'$ of
$v$ to have at least one incident edge in the data graph bearing the
same label of an edge incident on $v$ in the query graph. The score of
a candidate $v'$ is the similarity between the neighborhoods of $v$ and $v'$,
represented in the form of vectors, and further refined using an iterative
process. Finally, one unlabeled query
node is chosen as the pivot $p$. The top-$k$ candidates for multiple unlabeled
query nodes are put together to form answer tuples, if they are within the neighborhood
of $p$'s top-$k$ candidates. Similar to the best-first method (Sec.\ref{sec:processing}),
\system{Baseline} explores a query lattice in a bottom-up manner
and prunes ancestors of null nodes. However, differently, it evaluates
the lattice by breadth-first traversal instead of in the order of
upper-bound scores. There is no early-termination by \topk\ scores,
as \system{Baseline} terminates when every node is either evaluated or pruned.
We implemented \system{GQBE} and \system{Baseline} in Java and we obtained
the source code of \system{NESS} from the authors.
\reminder{say difference between NESS and GQBE}

%\system{NESS} is a graph querying framework that finds matches of query
%graphs with unlabeled nodes which correspond to query entity nodes in MQG.
%Note that \system{NESS} must take a query graph (instead of a query tuple)
%as input.  Hence we feed the MQG discovered by \system{GQBE} to \system{NESS}.
%\system{NESS} models the nodes in a data graph by feature vectors representing
%their neighboring nodes.  In the same way, the vectors for unlabeled nodes capture
%their neighboring nodes in the query graph.  The vectors for labeled nodes in
%the query graph are the same as their vectors in the data graph.  For each node
%$v$ in the query graph, a set of candidate nodes are identified.  The score of
%a candidate $v'$ is the similarity between the vectors of $v$ and $v'$.  Finding
%candidates thus resembles retrieving documents by intersecting postings lists for
%search keywords.  \system{NESS} does not consider edge-labeled graphs.  We adapted
%it by requiring $v'$ to have at least one incident edge in the data graph
%bearing the same label of an edge incident on $v$ in the query graph.  One unlabeled
%node is chosen as the pivot $p$. Candidates for multiple unlabeled
%nodes are put together to form answer tuples, if they are within the neighborhood
%of $p$'s candidate. An answer tuple's score is the vector
%similarity attained by that candidate.

\vspace{-1mm}
{\flushleft \textbf{Queries and Ground Truth}}\hspace{2mm}
Two groups of queries are used on the two datasets, respectively.
The Freebase queries F$_1$-- F$_{20}$ are based on Freebase tables
such as \textsf{\scriptsize http:// www.freebase.com/view/computer/programming\_language\_designer?instances},
except F$_1$ and F$_6$ which are from Wikipedia tables such as
\textsf{\scriptsize http://en.wikipedia.org/wiki/List\_of\_English\_football\_club\_owners}.
The DBpedia queries D$_1$-- D$_8$ are based on DBpedia tables such
as the values for property \textsf{\scriptsize is dbpedia-owl:author of}
on page \textsf{\scriptsize http://dbpedia.org/page/Microsoft}.
Each such table is a collection of tuples, in which each tuple consists
of one, two, or three entities.  For each table, we used one or more tuples as
query tuples and the remaining tuples as the ground truth for query answers.
All the $28$ queries and their corresponding table sizes are summarized in Table~\ref{tab:groundTruth}.
They cover diverse domains, including people, companies, movies,
sports, awards, religions, universities and automobiles.

\begin{table} [t]
\centering
\scriptsize
\begin{tabular}{|@{\hspace{0.5mm}}c@{\hspace{0.5mm}}|@{\hspace{0.5mm}}c@{\hspace{0.5mm}}|}
  \hline
  {\bf Query Tuple} & {\bf Top-$3$ Answer Tuples} \\
  \hline \hline
    & $\langle$D. Knuth, Stanford, V. Neumann Medal$\rangle$ \\
    $\langle$Donald Knuth, Stanford, Turing Award$\rangle$ & $\langle$J. McCarthy, Stanford, Turing Award$\rangle$ \\
    & $\langle$N. Wirth, Stanford, Turing Award$\rangle$ \\\hline
    & $\langle$David Filo, Yahoo!$\rangle$ \\
    $\langle$Jerry Yang, Yahoo!$\rangle$ & $\langle$Bill Gates, Microsoft$\rangle$ \\
    & $\langle$Steve Wozniak, Apple Inc.$\rangle$ \\\hline
    & $\langle$Java$\rangle$ \\
    $\langle$C$\rangle$ & $\langle$C++$\rangle$ \\
    & $\langle$C Sharp$\rangle$ \\
  \hline
\end{tabular}
\vspace{1mm}
\caption{Case Study: Top-$3$ Results for Selected Queries}
\label{tab:queryResults}
\vspace{-2mm}
\end{table}

\vspace{-2mm}
{\flushleft \textbf{Sample Answers}}\hspace{2mm} Table \ref{tab:queryResults} only
lists the top-$3$ results found by \system{GQBE} for $3$ queries (F$_1$, F$_{18}$, F$_{19}$),
due to space limitations.
%%%We observe that all the top-$3$ answer tuples are relevant with respect to the corresponding query tuple.

\spara{(A) Accuracy Based on Ground Truth}

We measured the accuracy of \system{GQBE} and \system{NESS} by
comparing their query results with the ground truth.
The accuracy on a set of queries is the average of accuracy on individual queries.
The accuracy on a single query is captured by three widely-used measures~\cite{Manning08}, as follows.
\begin{list}{$\bullet$}
{ \setlength{\leftmargin}{0.5em} }
\item Precision-at-$k$ (P@$k$): the percentage of the top-$k$ results
that belong to the ground truth.

\item Mean Average Precision (MAP): The average precision of the top-$k$
results is AvgP$=$$\frac{\sum_{i=1}^k{\text{P}@i}\ \times\ rel_i}{\text{size of ground truth}}$,
where $rel_i$ equals $1$ if the result at rank $i$ is in the ground truth
and $0$ otherwise.  MAP is the mean of AvgP for a set of queries.

\item Normalized Discounted Cumulative Gain (nDCG):
The cumulative gain of the top-$k$ results is
DCG$_k$$=$$rel_1$$+$$\sum_{i=2}^{k}{\frac{rel_i}{\log_2(i)}}$.
It penalizes the results if a ground truth result is ranked low.
DCG$_k$ is normalized by IDCG$_k$, the cumulative gain for an ideal ranking
of the top-$k$ results.  Thus nDCG$_k$$=$$\frac{\text{DCG}_k}{\text{IDCG}_k}$.
\end{list}

\begin{figure}[t]
\centering
\hspace{-5mm}
\subfigure[P@$k$]{
\includegraphics[scale=0.184, angle=270]{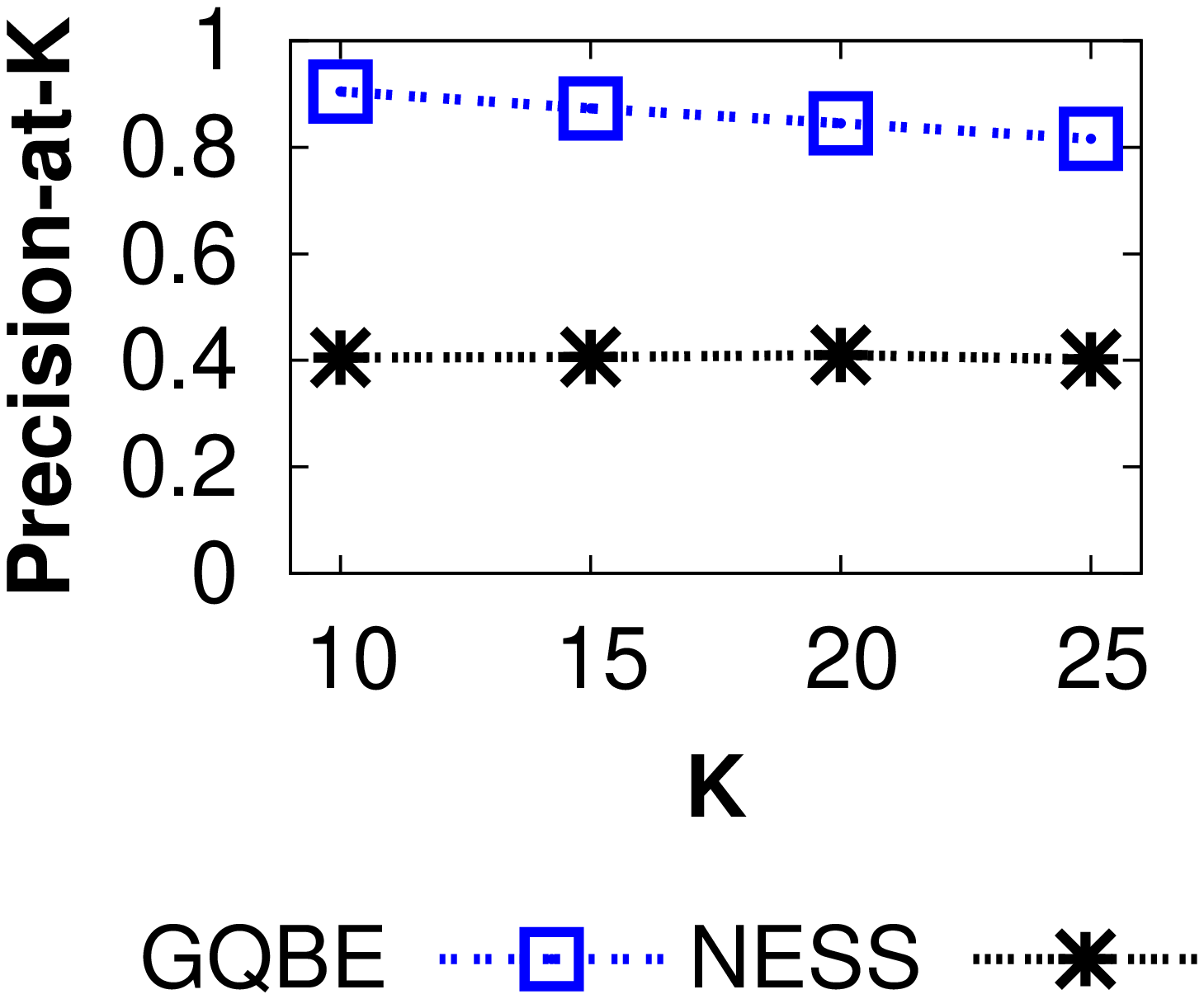}
\label{fig:precision}
}
\hspace{-6mm}
\subfigure[MAP]{
\includegraphics[scale=0.184, angle=270]{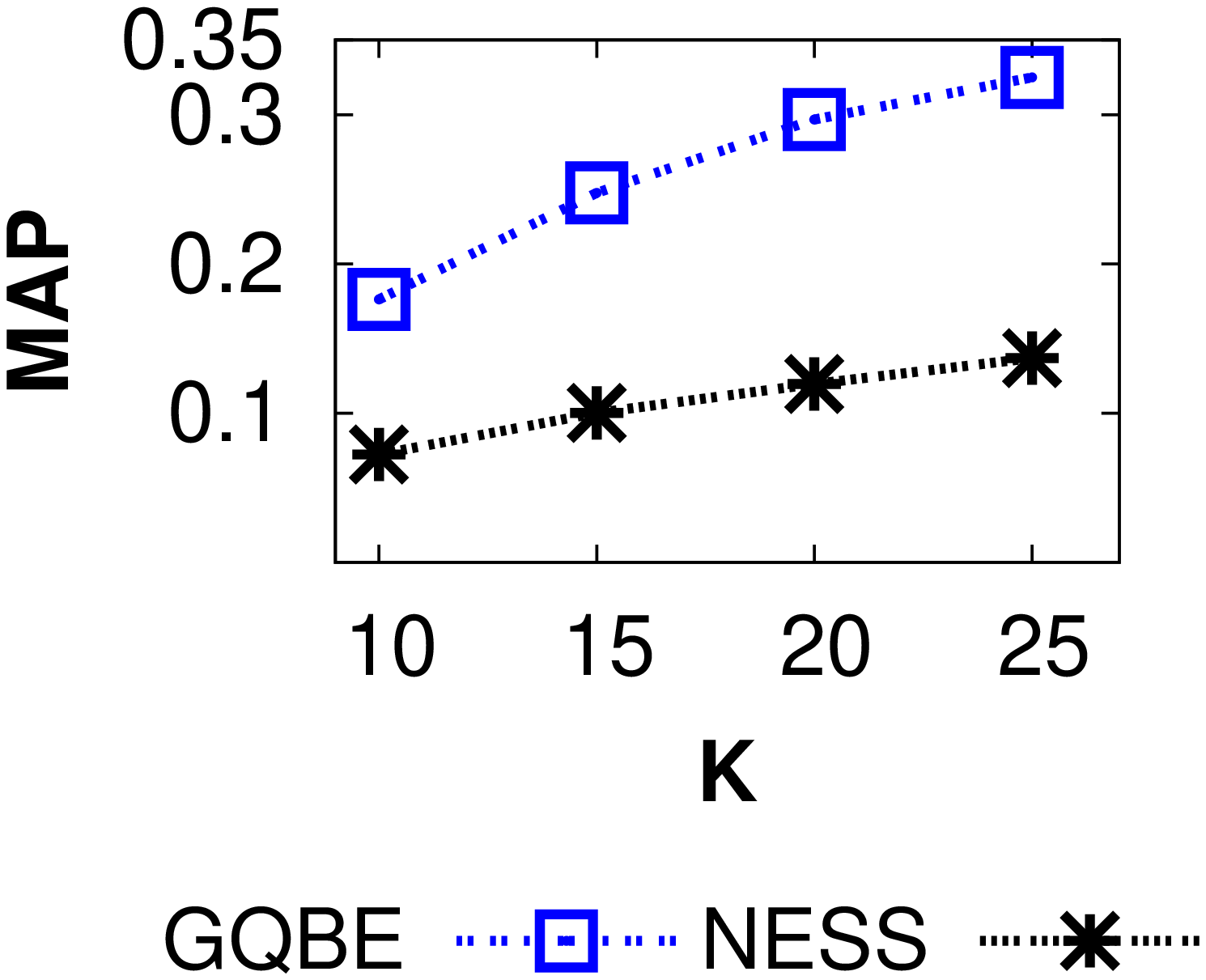}
\label{fig:map}
}
\hspace{-6mm}
\subfigure[nDCG]{
\includegraphics[scale=0.184, angle=270]{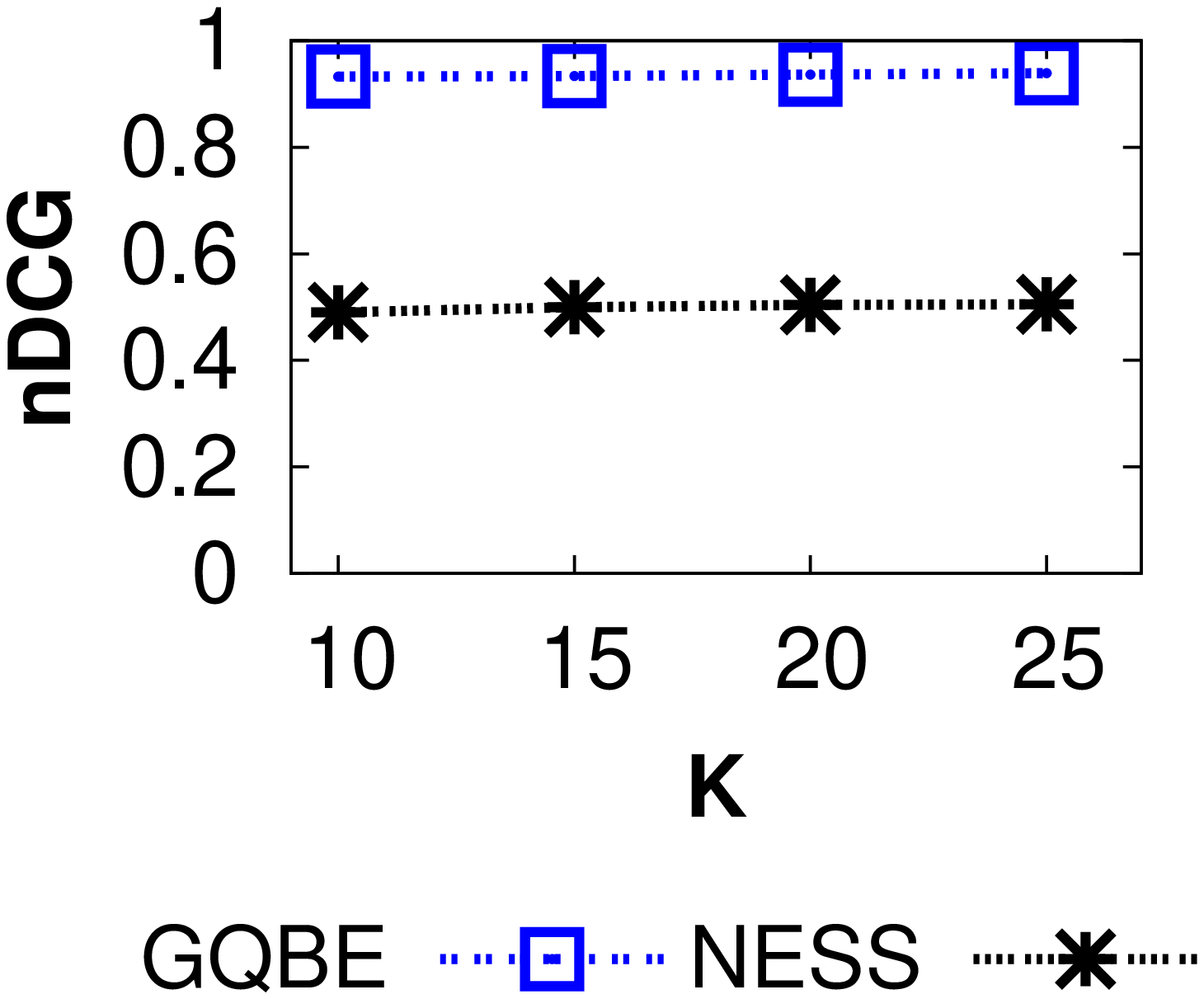}
\label{fig:ndcg}
}
\vspace{-1mm}
\caption{Accuracy of \system{GQBE} and \system{NESS} on Freebase Queries}
\label{fig:accuracy}
\end{figure}

\begin{table} [t]
\centering
\scriptsize
\begin{tabular}{|c|p{4mm}|c|c||c|p{4mm}|c|c|}
  \hline
  {\bf Query} & {\bf P@$k$} & {\bf nDCG} & {\bf AvgP} & {\bf Query} & {\bf P@$k$} & {\bf nDCG} & {\bf AvgP}\\
  \hline \hline
  D$_1$ & 1.00 & 1.00 & 0.20 & D$_2$ & 1.00 & 1.00 & 0.04 \\
  D$_3$ & 1.00 & 1.00 & 0.03 & D$_4$ & 0.80 & 0.94 & 0.19 \\
  D$_5$ & 0.90 & 1.00 & 0.08 & D$_6$ & 1.00 & 1.00 & 0.04 \\
  D$_7$ & 0.90 & 0.98 & 0.22 & D$_8$ & 1.00 & 1.00 & 0.01 \\
  \hline
 \end{tabular}
\vspace{1mm}
\caption{Accuracy of \system{GQBE} on DBpedia Queries, $k$=$10$}
\label{tab:dbpediaResults}
\vspace{-2mm}
\end{table}

Fig.\ref{fig:accuracy} shows these measures for different values
of $k$ on the Freebase queries.  \system{GQBE} has high accuracy.
For instance, its P@$25$ is over $0.8$.  The absolute value of MAP is not high, merely because
Fig.\ref{fig:map} only shows the MAP for at most top-$25$
results, while the ground truth size (i.e., the denominator in
calculating MAP) for many queries is much larger.
Moreover, \system{GQBE} outperforms \system{NESS} substantially, as its
accuracy in all three measures is almost always twice as better.
This is because \system{GQBE} gives priority to query entities and
important edges in MQG, while \system{NESS} gives equal
importance to all nodes and edges except the pivot.
Furthermore, the way \system{NESS} handles edge labels
does not explicitly require answer entities to be connected by
the same paths between query entities.
%%%The other reason is that \system{NESS} ranks answer graphs with respect to only one
%%%of the query entities.  In fact, \system{NESS} does not work well when a query graph
%%%has multiple nodes with unknown labels.  In contrast, \system{GQBE} is modeled to
%%%handle such cases, by allowing all the query entities to be replaced by
%%%answer entities.

Table~\ref{tab:dbpediaResults} further shows the accuracy of \system{GQBE}
on individual DBpedia queries at $k$=$10$.  It exhibits high accuracy on all queries,
including perfect precision in several cases.

\begin{table} [t]
\centering
\scriptsize
\begin{tabular}{|@{\hspace{1.5mm}}c@{\hspace{1.5mm}}|@{\hspace{1.5mm}}c@{\hspace{1.5mm}}||@{\hspace{1.5mm}}c@{\hspace{1.5mm}}|@{\hspace{1.5mm}}c@{\hspace{1.5mm}}||@{\hspace{1.5mm}}c@{\hspace{1.5mm}}|@{\hspace{1.5mm}}c@{\hspace{1.5mm}}||@{\hspace{1.5mm}}c@{\hspace{1.5mm}}|@{\hspace{1.5mm}}c@{\hspace{1.5mm}}|}
  \hline
  {\bf Query}  &   {\bf PCC}  & {\bf Query} & {\bf PCC} & {\bf Query} & {\bf PCC} & {\bf Query} & {\bf PCC} \\
  \hline
  F$_1$ & 0.79 & F$_2$ & 0.78 & F$_3$ & 0.60 & F$_4$ & 0.80 \\
  F$_5$ & 0.34 & F$_6$ & 0.27 & F$_7$ & 0.06 & F$_8$ & 0.26 \\
  F$_9$ & 0.33 & F$_{10}$ & 0.77 & F$_{11}$ & 0.58 & F$_{12}$ & undefined \\
  F$_{13}$ & undefined & F$_{14}$ & 0.62 & F$_{15}$ & 0.43 & F$_{16}$ & 0.29 \\
  F$_{17}$ & 0.64 & F$_{18}$ & 0.30 & F$_{19}$ & 0.40 & F$_{20}$ & 0.65 \\
  \hline
\end{tabular}
\vspace{1mm}
\caption{Pearson Correlation Coefficient (PCC) between \system{GQBE} and Amazon MTurk Workers, $k$=$30$}
\label{tab:pcc}
\vspace{-2mm}
\end{table}

\spara{(B) Accuracy Based on User Study}

We conducted an extensive user study through Amazon Mechanical Turk
(MTurk, \textsf{\scriptsize https://www.mturk.com/mturk/}) to evaluate \system{GQBE}'s
accuracy on Freebase queries, measured by Pearson Correlation Coefficient
(PCC).  For each of the $20$ queries, we obtained the top-$30$
answers from \system{GQBE} and generated $50$ random pairs of these answers.
We presented each pair to $20$ MTurk workers and asked for their preference
between the two answers in the pair.  Hence, in total, $20,000$ opinions were obtained.
We then constructed two value lists per query, $X$ and $Y$, which represent
\system{GQBE} and MTurk workers' opinions, respectively.  Each list has $50$
values, for the $50$ pairs.  For each pair, the value in $X$ is the difference between
the two answers' ranks given by \system{GQBE}, and the value in $Y$ is the difference
between the numbers of workers favoring the two answers.  The PCC value for a query
is $(\text{E}(XY)-\text{E}(X)\text{E}(Y))/(\sqrt{\text{E}(X^2)-(\text{E}(X))^2}\sqrt{\text{E}(Y^2)-(\text{E}(Y))^2})$.
The value indicates the degree of correlation between the pairwise ranking orders
produced by \system{GQBE} and the pairwise preferences given by MTurk workers.
The value range is from $-1$ to $1$.  A PCC value in the ranges of [$0.5$,$1.0$],
[$0.3$,$0.5$) and [$0.1$,$0.3$) indicates a strong, medium and small positive correlation,
respectively~\cite{pcc_cohen}.
PCC is undefined, by definition, when $X$ and/or $Y$ contain all equal values.

Table \ref{tab:pcc} shows the PCC values for F$_{1}$-F$_{20}$.  Out of the $20$ queries,
\system{GQBE} attained strong, medium and small positive correlation
with MTurk workers on $9$, $5$ and $3$ queries, respectively.  Only query F$_{7}$ shows no correlation.
%%%It can be observed that the ranking of answers
%%%is better for $3$-tuple queries, suggesting that having more query entities probably captures a better maximal
%%%query graph and thus produces better answer tuples.
Note that PCC is undefined for F$_{12}$ and F$_{13}$, because all the top-$30$ answer tuples
have the same score and thus the same rank, resulting in all zero values in $X$, i.e., \system{GQBE}'s list.
%%%This is likely to happen when all those answers are projected from answer graphs to the same query graph.
%%%The average PCC across all the $20$ queries is \emph{0.497}, which indicates a strong
%%%correlation between \system{GQBE} and MTurk workers.

\begin{table*} [t]
\centering
\scriptsize
\begin{tabular}{|p{5mm}|p{4mm}|p{5mm}|p{5mm}|p{4mm}|p{5mm}|p{5mm}|p{4mm}|p{5mm}|p{5mm}|p{4mm}|p{5mm}|p{5mm}|p{4mm}|p{5mm}|p{5mm}|}
  \hline
  {\bf Query} & \multicolumn{3}{|c|}{{\bf Tuple1}} & \multicolumn{3}{|c|}{{\bf Tuple2}} & \multicolumn{3}{|c|}{{\bf Combined (1,2)}}  & \multicolumn{3}{|c|}{{\bf Tuple3}} & \multicolumn{3}{|c|}{{\bf Combined (1,2,3)}}\\\hline
  \hline
   & P@$k$ & nDCG & AvgP & P@$k$ & nDCG & AvgP & P@$k$ & nDCG & AvgP & P@$k$ & nDCG & AvgP & P@$k$ & nDCG & AvgP\\\hline
  F$_1$ & {\bf 0.36} & 0.76 & 0.32 & {\bf 0.36} & {\bf 1.00} & {\bf 0.50} & 0.12 & 0.38 & 0.02 & {\bf 0.36} & 0.73 & 0.22 & 0.12 & 0.49 & 0.02 \\
  F$_2$ & 0.76 & {\bf 1.00} & 0.79 & 0.00 & 0.00 & 0.00 & {\bf 0.80} & {\bf 1.00} & 0.80 & 0.12 & 0.70 & 0.05 & {\bf 0.80} & {\bf 1.00} & {\bf 0.91} \\
  F$_4$ & 0.32 & 0.73 & 0.09 & 0.40 & 0.65 & 0.08 & {\bf 1.00} & {\bf 1.00} & {\bf 0.45} & N/A & N/A & N/A & N/A & N/A & N/A \\
  F$_6$ & 0.24 & 0.89 & 0.16 & 0.28 & 0.89 & 0.18 & {\bf 0.40} & 0.87 & 0.16 & 0.36 & 0.98 & {\bf 0.22} & 0.12 & {\bf 0.94} & 0.07 \\
  F$_8$ & 0.92 & 0.79 & 0.20 & {\bf 1.00} & {\bf 1.00} & {\bf 0.27} & 0.96 & 0.98 & 0.24 & 0.48 & 0.86 & 0.08 & {\bf 1.00} & {\bf 1.00} & {\bf 0.27} \\
  F$_9$ & 0.68 & 0.72 & 0.23 & 0.56 & 0.66 & 0.17 & 0.80 & 0.86 & 0.35 & {\bf 1.00} & {\bf 1.00} & 0.62 & {\bf 1.00} & {\bf 1.00} & {\bf 0.66} \\
  F$_{17}$ & 0.32 & {\bf 1.00} & 0.33 & 0.64 & 0.83 & 0.25 & 0.32 & {\bf 1.00} & 0.32 & 0.56 & 0.84 & 0.23 & {\bf 0.68} & {\bf 1.00} & {\bf 0.46} \\
  \hline
\end{tabular}
\vspace{1mm}
\caption{Accuracy of \system{GQBE} on Multi-tuple Queries, $k$=$25$}
\label{tab:multiTup}
\vspace{-3mm}
\end{table*}

%%%%%%%%%%%%%%%%%%%%%%%%%%%%%%%%%%%%%%%%%%%%
\begin{figure*}[htb]
\begin{minipage}[b]{0.5\linewidth}
\centering
  \includegraphics[width = 0.90\linewidth, keepaspectratio = true, scale=0.6]{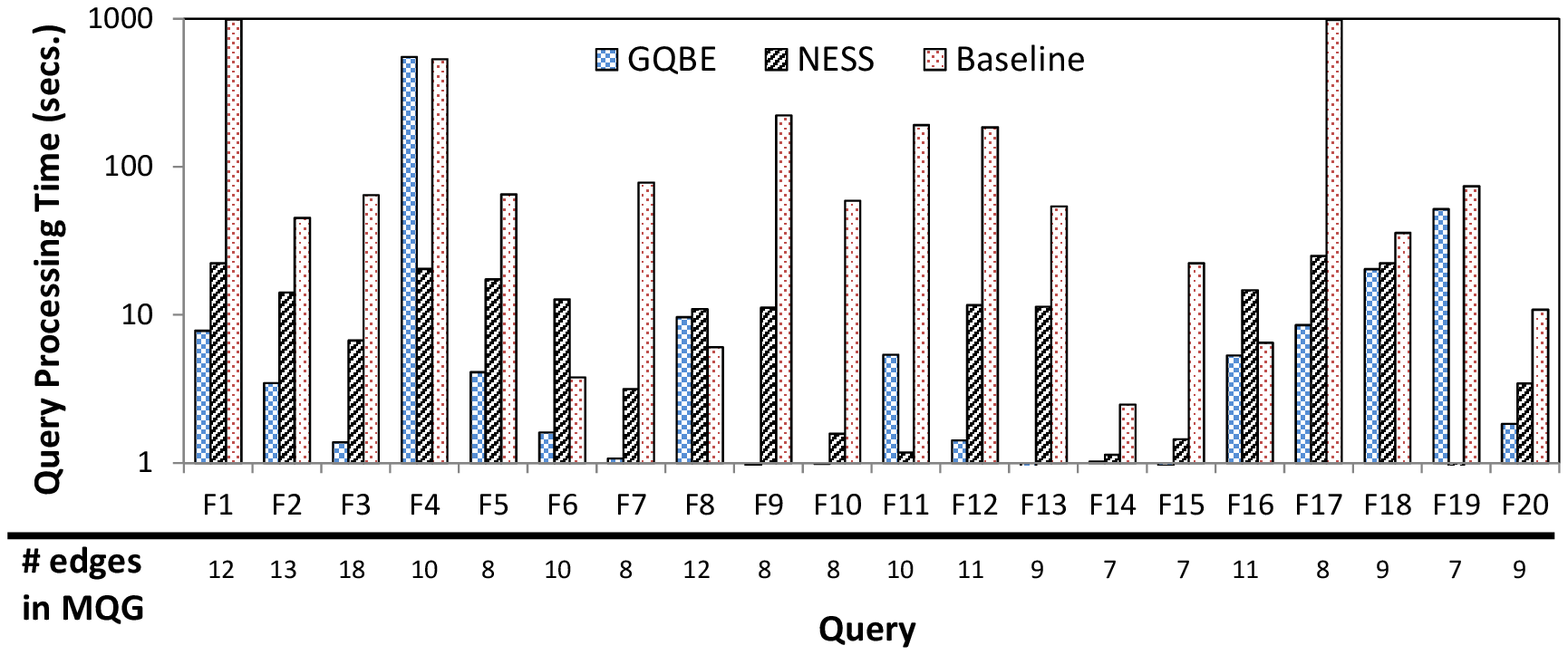}
\caption{Query Processing Time}
\label{fig:run-time}
\end{minipage}\vspace{-2mm}
\begin{minipage}[b]{0.5\linewidth}
\centering
  \includegraphics[width = 0.90\linewidth, keepaspectratio = true, scale=0.6]{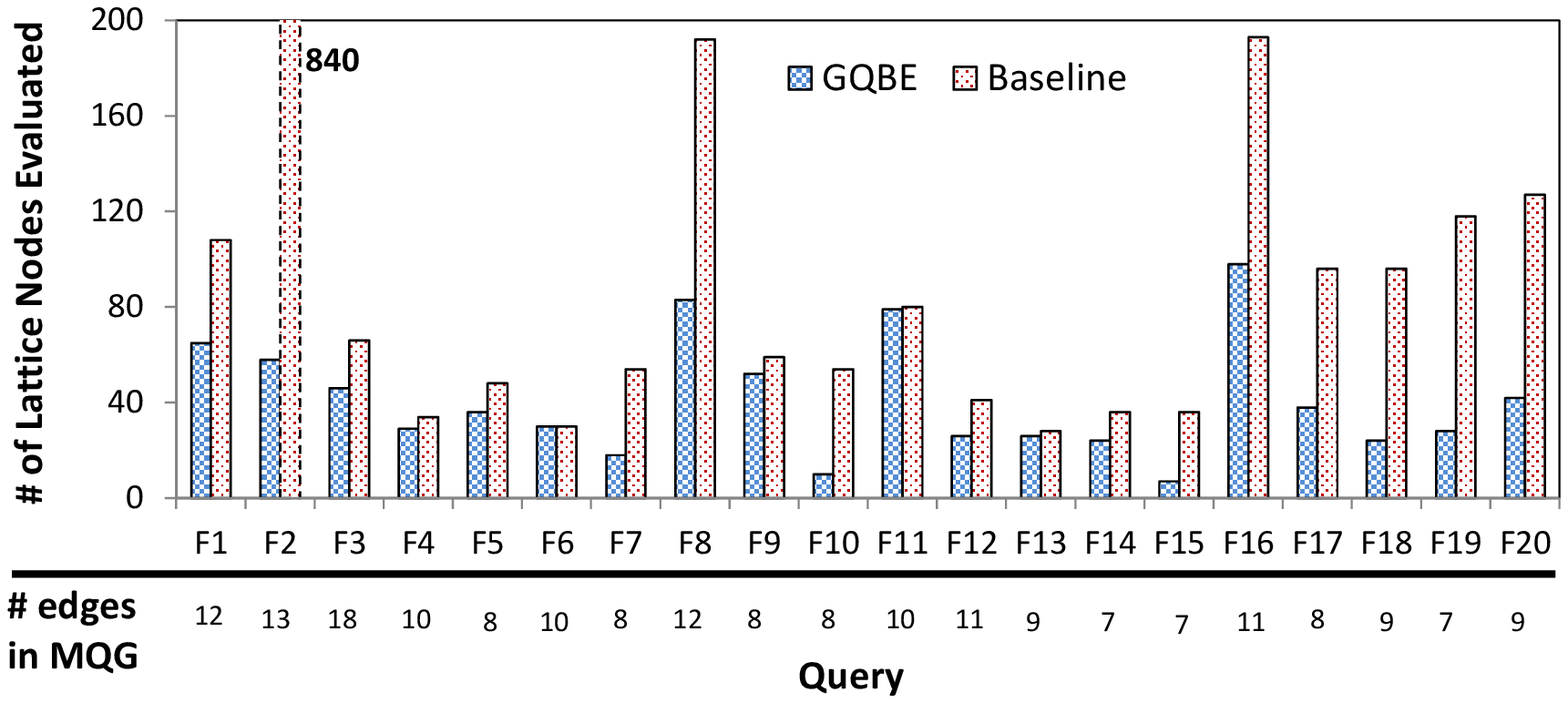}%\vspace{-3mm}
\caption{Lattice Nodes Evaluated}
\label{fig:nodes-traversed}
\end{minipage}\vspace{-2mm}
\end{figure*}

\spara{(C) Accuracy on Multi-tuple Queries}

We investigated the effectiveness of the multi-tuple querying approach
(Sec.\ref{sec:multituple}).  In aforementioned single-tuple query
experiment (A), \system{GQBE} attained perfect P@$25$ for $13$ of the
$20$ Freebase queries.  We thus focused on the remaining $7$ queries.
For each query, Tuple1 refers to the query tuple in
Table~\ref{tab:groundTruth}, while Tuple2 and Tuple3 are two tuples
from its ground truth.  Table~\ref{tab:multiTup} shows
the accuracy of top-$25$ \system{GQBE} answers for the three
tuples individually, as well as for the first two and three tuples
together by merged MQGs, which are denoted Combined(1,2) and
Combined(1,2,3), respectively.  F$_4$ attained perfect precision after
Tuple2 was included.  Therefore, Tuple3 was not used for F$_4$.  The
results show that, in most cases, Combined(1,2) had better accuracy
than individual tuples and Combined(1,2,3) further improved accuracy.

\spara{(D) Efficiency Results}

We compared the efficiency of \system{GQBE}, \system{NESS} and
\system{Baseline} on Freebase queries.  The total run time for a query
tuple is spent on two components---query graph discovery and query
processing.  Fig.\ref{fig:run-time} compares the three
methods' query processing time, in logarithmic scale.
The figure shows the query processing time for each of the $20$ Freebase queries, and the edge cardinality of the MQG for each of those is shown below the corresponding query id.
%For each edge cardinality of MQG,
%the figure shows the average time on queries with the same edge
%cardinality.  
The query cost does not appear to increase by edge
cardinality, regardless of the query method.
For \system{GQBE} and \system{Baseline}, this is because
query graphs are evaluated by joins and join selectivity plays a more
significant role in evaluation cost than number of edges.
\system{NESS} finds answers by intersecting postings lists on feature
vectors.  Hence, in evaluation cost, intersection size matters more
than edge cardinality.
\system{GQBE} outperformed \system{NESS} on $17$ of the $20$ queries
and was more than $3$ times faster in $10$ of them. It finished within $10$
seconds on $17$ queries. However, it performed very poorly 
on F$4$ and F$19$, which have $10$ and $7$ edges respectively.
%on a $9$-edge MQG
%and a $10$-edge MQG, taking $51$ and $552$ seconds, respectively.  
This indicates that the edges in the two MQGs lead to poor join selectivity.
%%%In average, \system{NESS} performed better on MQG with $9$ and $10$ edges.
\system{Baseline} clearly suffered, due to its inferior pruning
power compared to the best-first exploration employed by \system{GQBE}.
This is evident in Fig.\ref{fig:nodes-traversed} which shows the numbers
of lattice nodes evaluated for each query. 
%under varying edge cardinality of MQG.
\system{GQBE} evaluated considerably less nodes in most cases and
at least $2$ times less on $11$ of the $20$ queries.
%$6$ times less on queries with $12$ edges in MQG.  (The value for
%\system{Baseline} is $380$, which is off the chart and listed explicitly.)
%%%This can be explained by the fact that \system{NESS} might terminate query
%%%processing before finding the
%%%best score for an answer tuple, while \system{GQBE} always finds the
%%%best score an answer tuple can attain w.r.t an MQG. \textbf{unclear}

MQG discovery precedes the query processing step and is shared by all
three methods.  Column MQG$_1$ in Table~\ref{tab:multiTuple-mqg} lists
the time spent on discovering MQG for each Freebase query.  This time
component varies across individual queries, depending on the sizes of
query tuples' neighborhood graphs.  Compared to the values shown in
Fig.\ref{fig:run-time}, the time taken to discover an MQG in average
is comparable to the time spent by \system{GQBE} in evaluating it.

%%%%%%%%%%%%%%%%%%%%% alternative way to display the figures %%%%%%%%%%%%%%%%%%%%%%%%%%%%%%%%%
%\begin{figure*}[t]
%\centering
%\subfigure[Number of Lattice Nodes Evaluated]{
%\includegraphics[scale=0.25]{figures/latticeNodesTraversed.eps}
%\label{fig:nodes-traversed}
%}
%\hspace{0.1cm}
%\subfigure[Execution Time]{
%\includegraphics[scale=0.29]{figures/executionTime.eps}
%\label{fig:run-time}
%}
%\hspace{0.1cm}
%\subfigure[Execution Time of Multi-tuple Queries]{
%\includegraphics[scale=0.31]{figures/MultiTupleLatticeEvaluation.eps}
%\label{fig:multiTuple-exec-time}
%}
%\vspace{-1mm}
%\caption{Efficiency Results}
%\label{fig:efficiency}
%\vspace{-4mm}
%\end{figure*}

\begin{table} [t]
\centering
\scriptsize
\begin{tabular}{|@{\hspace{1.5mm}}c@{\hspace{1.5mm}}|@{\hspace{1.5mm}}c@{\hspace{1.5mm}}|@{\hspace{1.5mm}}c@{\hspace{1.5mm}}|@{\hspace{1.5mm}}c@{\hspace{1.5mm}}||@{\hspace{1.5mm}}c@{\hspace{1.5mm}}|@{\hspace{1.5mm}}c@{\hspace{1.5mm}}|@{\hspace{1.5mm}}c@{\hspace{1.5mm}}|@{\hspace{1.5mm}}c@{\hspace{1.5mm}}|}
  \hline
  {\bf Query}  &   {\bf MQG$_1$}  & {\bf MQG$_2$} & {\bf Merge} & {\bf Query} & {\bf MQG$_1$}  & {\bf MQG$_2$} & {\bf Merge} \\
  \hline
  F$_1$ & 73.141 & 73.676 & 0.034 & F$_2$ & 0.049 & 0.029 & 0.006 \\
  F$_3$ & 12.566 & 4.414 & 0.024 & F$_4$ & 5.731 & 7.083 & 0.024 \\
  F$_5$ & 9.982 & 2.522 & 0.079 & F$_6$ & 6.082 & 4.654 & 0.039 \\
  F$_7$ & 0.152 & 0.107 & 0.007 & F$_8$ & 10.272 & 2.689 & 0.032 \\
  F$_9$ & 62.285 & 2.384 & 0.041 & F$_{10}$ & 2.910 & 5.933 & 0.030 \\
  F$_{11}$ & 59.541 & 65.863 & 0.032 & F$_{12}$ & 1.977 & 0.021 & 0.006 \\
  F$_{13}$ & 9.481 & 5.624 & 0.034 & F$_{14}$ & 0.038 & 0.015 & 0.004 \\
  F$_{15}$ & 0.154 & 5.143 & 0.021 & F$_{16}$ & 54.870 & 6.928 & 0.057 \\
  F$_{17}$ & 60.582 & 69.961 & 0.041 & F$_{18}$ & 58.807 & 75.128 & 0.053 \\
  F$_{19}$ & 0.224 & 0.076 & 0.003 & F$_{20}$ & 0.025 & 0.017 & 0.002 \\
  \hline
\end{tabular}
\vspace{1mm}
\caption{Time for Discovering and Merging MQGs (secs.)}
\label{tab:multiTuple-mqg}
\vspace{-2mm}
\end{table}

\begin{figure}[t]
\centering
  \includegraphics[width = 0.5\linewidth, keepaspectratio = true, scale=0.4]{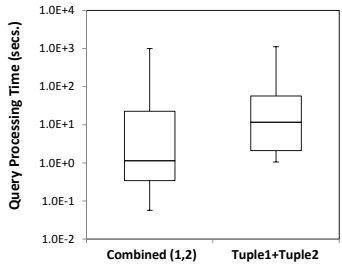}\vspace{-3mm}
\caption{Query Processing Time of $2$-tuple Queries}
\label{fig:multiTuple-exec-time}
\end{figure}

%\begin{figure}[t]
%\centering
%%%%\includegraphics[scale=0.18, angle=270]{figures/multiTupleLogScale.eps}
%\includegraphics[scale=0.5]{figures/1.eps}
%%\vspace{-3mm}
%\caption{Execution Time of Multi-tuple Queries}
%\label{fig:multiTuple-time}
%%\vspace{-4mm}
%\end{figure}
%%% A table for the following figure exists..
%\begin{figure}[t]
%\centering
%%%%\includegraphics[scale=0.18, angle=270]{figures/multiTupleLogScale.eps}
%\includegraphics[scale=0.5]{figures/MultiTupleMQGCreation.eps}
%%\vspace{-3mm}
%\caption{MQG Creation Time of Multi-tuple Queries}
%\label{fig:multiTuple-mqg-time}
%%\vspace{-4mm}
%\end{figure}
%%%%%%%%%%%%%%%%%%%%%%%%%%%%%%%%%%%%%%%%%%%%%%%%%%%%%%%%%%%%%%
%\begin{figure}[t]
%\centering
%%%%\includegraphics[scale=0.18, angle=270]{figures/multiTupleLogScale.eps}
%\includegraphics[scale=0.35]{figures/MultiTupleLatticeEvaluation.eps}
%%\vspace{-3mm}
%\caption{Execution Time of Multi-tuple Queries}
%\label{fig:multiTuple-exec-time}
%%\vspace{-4mm}
%\end{figure}
%%%%%%%%%%%%%%%%%%%%%%%%%%%%%%%%%%%%%%%%%%%%%%%%%%%%%%%%%%

Fig.\ref{fig:multiTuple-exec-time} shows the distribution of the \system{GQBE}'s query
processing time, in logarithmic scale, on the merged
MQGs of 2-tuple queries in Table~\ref{tab:multiTup}, denoted by
Combined(1,2).  It also shows the distribution of the total time for evaluating the two
tuples' MQGs individually, denoted Tuple1+Tuple2. 
Combined(1,2) processes $10$ of the $20$ queries in less than a second while the fastest query for Tuple1+Tuple2 takes a second.
%The time for
%Combined(1,2) is $1$-$3$ orders of magnitude less in $8$ out of
%$20$ queries and is significantly less in $5$ other queries.
This suggests that the merged MQGs gave higher weights to more
selective edges, resulting in faster lattice evaluation.
Meanwhile, these selective edges are also more important edges
common to the two query tuples, leading to improved answer accuracy
as shown in Table~\ref{tab:multiTup}.  Table~\ref{tab:multiTuple-mqg}
further shows the time taken to discover MQG$_1$ and MQG$_2$
for the two tuples, along with the time for merging them.
The latter is negligible compared to the former.

\section{Related Work}\label{sec:related}

Lim et al.~\cite{lim_edbt13} use example tuples to find similar tuples in database tables that are coupled with ontologies.
%The examples here are exact tuples of a table, and other similar tuples are found.
They do not deal with graph data and the example tuples are not formed by entities.
%%%In contrast, \system{GQBE} considers a set of related entities as example tuples, whose relationship can be established by joins between multiple tables.
%\cite{Abouzied:2013} provides a theoretical aspect of the example-driven query specification problem. Users are provided with answer tuples, and feedback on the relevance of each answer tuple is used to refine the query. This work only deals with a special class of quantified boolean queries called \emph{qhorn}.
%%%\system{GQBE} also finds approximate answers to the hidden query graph derived from the input tuple.

The goal of \emph{set expansion} is to grow a set of objects starting from seed objects.  Example systems include \cite{seal}, \cite{Gupta-2009}, and the now defunct \system{Google\ Sets} and \system{Squared} services
%\footnote {\scriptsize http:// en.wikipedia.org/wiki/List\_of\_Google\_products}
(\textsf{\scriptsize http://en.wikipedia.org/wiki/List\_of\_Google\_products}).
Chang et al.~\cite{infoNebula} identify \topk\ correlated keyword terms from an information network given a set of terms, where each term can be an entity.  These systems, except~\cite{infoNebula}, do not operate on data graphs.  Instead, they find existing answers within structures in web pages such as HTML tables and lists.  Furthermore, all these systems except \system{Google\ Squared} and \cite{Gupta-2009} take a set of individual entities as input.  \system{GQBE} is more general in that each query tuple contains multiple entities.

Several works~\cite{TF06,Kasneci+09ming,Fang+11rex} %~\cite{TF06,Kasneci+09ming,MG10,Fang+11rex}
identify the best subgraphs/paths in a data graph to describe how several input nodes are related.   The query graph discovery component of \system{GQBE} is different in important ways-- (1) The graphs in~\cite{TF06} %,MG10}
contain nodes of the same type and edges representing the same relationship, e.g., social networks capturing friendship between people. %%% and co-authorship graphs capturing co-authoring relation between authors.
%Graphs with only a few different types of entities and relationships are also similar.
The graphs in \system{GQBE} and others~\cite{Kasneci+09ming,Fang+11rex} have many different types of entities and relationships. (2) The paths discovered by their techniques only connect the input nodes.  \system{REX}~\cite{Fang+11rex} has the further limitation of allowing only two input entities.  Differently the maximal query graph in \system{GQBE} includes edges incident on individual query entities.  (3) \system{GQBE} uses the discovered query graph to find answer graphs and answer tuples, which is not within the focus of the aforementioned works.

%%%The query processing component of \system{GQBE} is related to the problem of
%%%inexact subgraph matching, which identifies all approximate occurrences of a
%%%query graph in a data graph.
%%%In bioinformatics, exact and approximate \emph{subgraph matching} have been
%%%extensively studied, e.g., \system{PathBlast} \cite{KYLSSI04}, \system{SAGA}
%%%\cite{saga}, \system{NetAlign} \cite{LXTN06}, \system{IsoRank} \cite{SXB08}.

There are many studies on approximate/inexact subgraph matching in
large graphs, such as \system{G}-\system{Ray}~\cite{TFGE07}, \system{TALE}~\cite{tale} and
\system{NESS}~\cite{ness}. %%%(See \cite{Galla06, KhanWY12} for surveys.)
\system{GQBE}'s query processing component is different from them on several aspects.
First, \system{GQBE} only requires to match edge labels and matching node identifiers is
not mandatory.  This is equivalent to matching a query graph
with all unlabeled nodes and thereby significantly increases the problem
complexity.  Only a few previous methods (e.g., \system{NESS}~\cite{ness})
allow unlabeled query nodes.  Second, in \system{GQBE}, the \topk\ query algorithm centers
around query entities.
More specifically, the weighting function gives more importance
to query entities and the minimal query trees mandate the presence of entities
corresponding to query entities.
%More specifically, edges in the maximal query
%graph are weighted by their distances to query entities, and an answer
%graph must have entities corresponding to all query entities.
On the contrary, previous methods give equal importance to all nodes in a query graph,
since the notion of query entity does not exist there.
Our empirical results show that this difference makes \system{NESS} produce less
accurate answers than \system{GQBE}.  Finally, although the query relaxation DAG proposed
in~\cite{sihem-Yahia05} is similar to \system{GQBE}'s query lattice, the scoring mechanism
of their relaxed queries is different and depends on XML-based relaxations.

\section{Conclusion}\label{sec:conclude}

We introduce \system{GQBE}, a system that queries
knowledge graphs by example entity tuples.
As an initial step toward better usability of graph query systems,
\system{GQBE} saves users the burden of forming explicit query graphs.
To the best of our knowledge, there has been no such proposal in the past.
Its query graph discovery component derives a hidden query graph based on example tuples.
The query lattice based on this hidden graph may contain a large number of query graphs.
\system{GQBE}'s query algorithm only partially evaluates query graphs for
obtaining the \topk\ answers.
Experiments on Freebase and DBpedia datasets show that \system{GQBE}
outperforms the state-of-the-art system \system{NESS} on both accuracy and efficiency.
%Other experiments using Freebase and DBpedia
%show the effectiveness of our system with the querying paradigm proposed.
%%%As we witness an unprecedented proliferation of entity data graphs in the real
%%%world, we hope this work will have profound impact on many applications.

\vspace{-1mm}
\section{Acknowledgements}\label{sec:acknowledgement}
The authors would like to thank Mahesh Gupta for his valuable contributions in performing the experiments.

\small
\bibliographystyle{abbrv}
\bibliography{gqbe}	

\begin{thebibliography}{10}

\bibitem{abadi07}
D.~J. Abadi, A.~Marcus, S.~Madden, and K.~J. Hollenbach.
\newblock Scalable semantic web data management using vertical partitioning.
\newblock In {\em VLDB'07}.

\bibitem{sihem-Yahia05}
S.~Amer-Yahia, N.~Koudas, A.~Marian, D.~Srivastava, and D.~Toman.
\newblock Structure and content scoring for xml.
\newblock In {\em VLDB}, 2005.

\bibitem{AuerBK+07}
S.~Auer, C.~Bizer, G.~Kobilarov, J.~Lehmann, R.~Cyganiak, and Z.~Ives.
\newblock {DB}pedia: A nucleus for a {Web} of open data.
\newblock In {\em ISWC}, 2007.

\bibitem{Bollacker+08freebase}
K.~Bollacker, C.~Evans, P.~Paritosh, T.~Sturge, and J.~Taylor.
\newblock Freebase: a collaboratively created graph database for structuring
  human knowledge.
\newblock In {\em SIGMOD}, pages 1247--1250, 2008.

\bibitem{infoNebula}
L.~Chang, J.~X. Yu, L.~Qin, Y.~Zhu, and H.~Wang.
\newblock Finding information nebula over large networks.
\newblock In {\em CIKM}, 2011.

\bibitem{GRAPHITE}
D.~H. Chau, C.~Faloutsos, H.~Tong, J.~I. Hong, B.~Gallagher, and
  T.~Eliassi-Rad.
\newblock {GRAPHITE}: A visual query system for large graphs.
\newblock In {\em ICDM Workshops}, pages 963--966, 2008.

\bibitem{pcc_cohen}
J.~Cohen.
\newblock {\em {Statistical Power Analysis for the Behavioral Sciences}}.
\newblock Lawrence Erlbaum Associates, 1988.

\bibitem{Demidova+12}
E.~Demidova, X.~Zhou, and W.~Nejdl.
\newblock {FreeQ}: an interactive query interface for {Freebase}.
\newblock In {\em WWW}, demo paper, 2012.

\bibitem{Fang+11rex}
L.~Fang, A.~D. Sarma, C.~Yu, and P.~Bohannon.
\newblock {REX}: explaining relationships between entity pairs.
\newblock In {\em PVLDB}, pages 241--252, 2011.

\bibitem{GabowM78}
H.~N. Gabow and E.~W. Myers.
\newblock Finding all spanning trees of directed and undirected graphs.
\newblock {\em SIAM J. Comput.}, 7(3):280--287, 1978.

\bibitem{Gupta-2009}
R.~Gupta and S.~Sarawagi.
\newblock Answering table augmentation queries from unstructured lists on the
  web.
\newblock In {\em VLDB}, pages 289--300, 2009.

\bibitem{usability}
H.~V. Jagadish, A.~Chapman, A.~Elkiss, M.~Jayapandian, Y.~Li, A.~Nandi, and
  C.~Yu.
\newblock Making database systems usable.
\newblock In {\em SIGMOD}, 2007.

\bibitem{Jarrar+12}
M.~Jarrar and M.~D. Dikaiakos.
\newblock A query formulation language for the data web.
\newblock {\em TKDE}, 24:783--798, 2012.

\bibitem{gqbedemo}
N.~Jayaram, M.~Gupta, A.~Khan, C.~Li, X.~Yan, and R.~Elmasri.
\newblock {GQBE}: Querying knowledge graphs by example entity tuples.
\newblock In {\em ICDE (demo description)}, 2014 (to appear).

\bibitem{GBLENDER}
C.~Jin, S.~S. Bhowmick, X.~Xiao, J.~Cheng, and B.~Choi.
\newblock {GBLENDER}: Towards blending visual query formulation and query
  processing in graph databases.
\newblock In {\em SIGMOD}, pages 111--122, 2010.

\bibitem{KA11}
M.~Kargar and A.~An.
\newblock Keyword search in graphs: Finding r-cliques.
\newblock {\em PVLDB}, pages 681--692, 2011.

\bibitem{Kasneci+09ming}
G.~Kasneci, S.~Elbassuoni, and G.~Weikum.
\newblock {MING}: mining informative entity relationship subgraphs.
\newblock In {\em CIKM}, 2009.

\bibitem{ness}
A.~Khan, N.~Li, X.~Yan, Z.~Guan, S.~Chakraborty, and S.~Tao.
\newblock Neighborhood based fast graph search in large networks.
\newblock In {\em SIGMOD'11}.

\bibitem{KhanWY12}
A.~Khan, Y.~Wu, and X.~Yan.
\newblock Emerging graph queries in linked data.
\newblock In {\em ICDE}, pages 1218--1221, 2012.

\bibitem{LZZC09}
Z.~Li, S.~Zhang, X.~Zhang, and L.~Chen.
\newblock Exploring the constrained maximum edge-weight connected graph
  problem.
\newblock {\em Acta Mathematicae Applicatae Sinica}, 25:697--708, 2009.

\bibitem{lim_edbt13}
L.~Lim, H.~Wang, and M.~Wang.
\newblock Semantic queries by example.
\newblock In {\em Proceedings of the 16th International Conference on Extending
  Database Technology (EDBT 2013)}, 2013.

\bibitem{spark}
Y.~Luo, X.~Lin, W.~Wang, and X.~Zhou.
\newblock Spark: top-k keyword query in relational databases.
\newblock In {\em SIGMOD}, 2007.

\bibitem{Manning08}
C.~D. Manning, P.~Raghavan, and H.~Schtze.
\newblock {\em Introduction to Information Retrieval}.
\newblock Cambridge University Press, NY, USA, 2008.

\bibitem{PoundIW10}
J.~Pound, I.~F. Ilyas, and G.~E. Weddell.
\newblock Expressive and flexible access to web-extracted data: a keyword-based
  structured query language.
\newblock In {\em SIGMOD}, pages 423--434, 2010.

\bibitem{SuchanekKW07}
F.~M. Suchanek, G.~Kasneci, and G.~Weikum.
\newblock {YAGO}: a core of semantic knowledge unifying {WordNet} and
  {Wikipedia}.
\newblock In {\em WWW}, 2007.

\bibitem{Sun+11}
Y.~Sun, J.~Han, X.~Yan, P.~S. Yu, , and T.~Wu.
\newblock {PathSim}: Meta path-based top-k similarity search in heterogeneous
  information networks.
\newblock {\em VLDB}, 2011.

\bibitem{tale}
Y.~Tian and J.~M. Patel.
\newblock {TALE}: A tool for approximate large graph matching.
\newblock In {\em ICDE}, pages 963--972, 2008.

\bibitem{TF06}
H.~Tong and C.~Faloutsos.
\newblock Center-piece subgraphs: Problem definition and fast solutions.
\newblock In {\em KDD}, pages 404--413, 2006.

\bibitem{TFGE07}
H.~Tong, C.~Faloutsos, B.~Gallagher, and T.~Eliassi-Rad.
\newblock Fast best-effort pattern matching in large attributed graphs.
\newblock {\em KDD}, 2007.

\bibitem{seal}
R.~C. Wang and W.~W. Cohen.
\newblock Language-independent set expansion of named entities using the web.
\newblock In {\em ICDM}, pages 342--350, 2007.

\bibitem{probase}
W.~Wu, H.~Li, H.~Wang, and K.~Q. Zhu.
\newblock Probase: a probabilistic taxonomy for text understanding.
\newblock In {\em SIGMOD}, pages 481--492, 2012.

\bibitem{Yahya+12}
M.~Yahya, K.~Berberich, S.~Elbassuoni, M.~Ramanath, V.~Tresp, and G.~Weikum.
\newblock Deep answers for naturally asked questions on the web of data.
\newblock In {\em WWW}, demo paper, pages 445--449, 2012.

\bibitem{YCHH12}
J.~Yao, B.~Cui, L.~Hua, and Y.~Huang.
\newblock Keyword query reformulation on structured data.
\newblock {\em ICDE}, pages 953--964, 2012.

\bibitem{YSZH12}
X.~Yu, Y.~Sun, P.~Zhao, and J.~Han.
\newblock Query-driven discovery of semantically similar substructures in
  heterogeneous networks.
\newblock In {\em KDD'12}.

\bibitem{qbe}
M.~M. Zloof.
\newblock Query by example.
\newblock In {\em AFIPS}, 1975.

\end{thebibliography}

\end{document}